\documentclass[11pt]{article}
\usepackage{fullpage}
\usepackage{amsmath,amsthm,amsfonts,amssymb,amscd}
\usepackage{times}
\usepackage{enumitem}
\usepackage{color}
\usepackage{authblk}
\usepackage{thm-restate}
\usepackage
[
colorlinks=true,
urlcolor=blue,
linkcolor=blue,
citecolor=blue,
]
{hyperref}

\newtheorem{lemma}{Lemma}[section]
\newtheorem{theorem}[lemma]{Theorem}

\newtheorem{definition}[lemma]{Definition}
\newtheorem{proposition}[lemma]{Proposition}

\newtheorem{claim}[lemma]{Claim}

\newtheorem{oprob}[lemma]{Open problem}
\newtheorem{remark}[lemma]{Remark}
\numberwithin{lemma}{section}

\newtheorem{informal theorem}[lemma]{Informal Theorem}

\newcommand{\set}[1]{\left \{{#1} \right \}}
\newcommand{\poly}{{\operatorname{poly}}}
\newcommand{\remove}[1]{}

\newcommand{\B}{\{0,1\}}
\newcommand{\ar}{\rightarrow}
\newcommand{\remph}[1]{\textsf{#1}}
\newcommand{\eps}{\epsilon}

\newcommand{\from}{\leftarrow}

\newcommand{\Hi}{H_{\infty}}

\newcommand{\Exp}{\mathbb{E}}

\newcommand{\half}{\frac{1}{2}}

\newcommand{\Con}{\operatorname{Con}}
\newcommand{\Red}{\operatorname{Red}}
\newcommand{\Enc}{\operatorname{Enc}}
\newcommand{\Dec}{\operatorname{Dec}}
\newcommand{\shared}{\operatorname{shared}}
\newcommand{\List}{\operatorname{List}}
\newcommand{\BSC}{\operatorname{BSC}}
\newcommand{\Hadamard}{\operatorname{Had}}
\newcommand{\RNSY}{\operatorname{RNSY}}
\newcommand{\NSY}{\operatorname{NSY}}

\newcommand{\Fix}{\operatorname{Fix}}
\newcommand{\Cond}{\operatorname{Cond}}

\newcommand{\MU}{\operatorname{MDIST}}
\newcommand{\MD}{\operatorname{MD}}
\newcommand{\WBV}{\operatorname{WBV}}

\newcommand{\newowf}{\operatorname{newOWF}}
\newcommand{\pred}{\operatorname{pred}}
\newcommand{\weight}{\operatorname{weight}}

\newcommand{\D}{D}

\title{Query complexity lower bounds \\ for local list-decoding and hard-core predicates \\ (even for small rate and huge lists)\footnote{A preliminary version of this manuscript appeared in the proceedings of ITCS 2021~\cite{RonZewiSV21}.}}

\author[1]{Noga Ron-Zewi\footnote{\texttt{noga@cs.haifa.ac.il}}}
\author[1]{Ronen Shaltiel\footnote{\texttt{ronen@cs.haifa.ac.il}}}
\author[2]{Nithin Varma\footnote{\texttt{nithvarma@gmail.com}}}

\affil[1]{\small University of Haifa, Israel}
\affil[2]{\small Max Planck Institute for Informatics, Saarland Informatics Campus, Germany}

\begin{document}

\setlist{itemsep=0\baselineskip}

\makeatletter
\renewcommand{\paragraph}{%
  \@startsection{paragraph}{4}%
  {\z@}{2ex \@plus 1ex \@minus .2ex}{-1em}%
  {\normalfont\normalsize\bfseries}%
}
\makeatother
\date{}
\begin{titlepage}

\maketitle
\thispagestyle{empty}

\begin{abstract}
A binary code $\Enc:\B^k \ar \B^n$ is $(\half-\eps,L)$-\emph{list decodable} if for every $w \in \B^n$, the set $\List(w)$, containing all messages $m \in \B^k$ such that the relative Hamming distance between $\Enc(m)$ and $w$ is at most $\half-\eps$, is of size at most $L$.
Informally, a $q$-query \emph{local list-decoder} for $\Enc$ is a randomized procedure $\Dec:[k]\times [L] \to \{0,1\}$ that when given oracle access to a string $w$, makes at most $q$ oracle calls, and for every message $m \in \List(w)$, with high probability, there exists $j \in [L]$ such that for every $i \in [k]$, with high probability, $\Dec^w(i,j)=m_i$.

We prove lower bounds on $q$, that apply even if $L$ is huge (say $L=2^{k^{0.9}}$) and the rate of $\Enc$ is small (meaning that $n \ge 2^{k}$):
\begin{itemize}
\item For $\eps \geq 1/k^{\nu}$ for some universal constant $0< \nu < 1$, we prove a lower bound of $q=\Omega(\frac{\log(1/\delta)}{\eps^2})$, where $\delta$ is the error probability of the local list-decoder. This bound is tight as there is a matching upper bound by Goldreich and Levin (STOC 1989) of $q=O(\frac{\log(1/\delta)}{\eps^2})$ for the Hadamard code (which has $n=2^k$). This bound extends an earlier work of Grinberg, Shaltiel and Viola (FOCS 2018) which only works if $n \le 2^{k^{\gamma}}$ for some universal constant $0<\gamma <1$, and the number of coins tossed by $\Dec$ is small (and therefore does not apply to the Hadamard code, or other codes with low rate).

\item For smaller $\eps$, we prove a lower bound of roughly $q = \Omega(\frac{1}{\sqrt{\eps}})$. To the best of our knowledge, this is the first lower bound on the number of queries of local list-decoders that gives $q \ge k$ for small~$\eps$.
\end{itemize}

Local list-decoders with small $\eps$ form the key component in the celebrated theorem of Goldreich and Levin that extracts a \emph{hard-core predicate} from a one-way function.
We show that black-box proofs \emph{cannot} improve the Goldreich-Levin theorem and produce a hard-core predicate that is hard to predict with probability $\half+\frac{1}{\ell^{\omega(1)}}$ when provided with a one-way function $f:\B^{\ell} \ar \B^{\ell}$, where $f$ is such that circuits of size $\poly(\ell)$ cannot invert $f$ with probability $\rho=1/2^{\sqrt{\ell}}$ (or even $\rho=1/2^{\Omega(\ell)}$). This limitation applies to any proof by black-box reduction (even if the reduction is allowed to use nonuniformity and has oracle access to $f$).
%
\end{abstract}
\end{titlepage}

\section{Introduction}
\label{sec:intro}

We prove limitations on local list-decoding algorithms and on reductions establishing hard-core predicates.

\subsection{Locally list-decodable codes}

List-decodable codes are a natural extension of (uniquely decodable) error-correcting codes, as it allows (list) decoding for error regimes where unique decoding is impossible. This is an extensively studied area; see \cite{Gur-survey06} for a survey. In this paper, we will be interested in list-decoding of binary codes.

\begin{definition}[List-decodable code]
For a function $\Enc:\B^k \ar \B^n$, and $w \in \B^n$, we define
\[ \List^{\Enc}_{\alpha}(w)=\set{m \in \B^k: \mathsf{dist}(\Enc(m),w) \le \alpha}.\footnote{For two strings $x,y \in \B^n$ we use $\mathsf{dist}(x,y)$ to denote the \emph{relative Hamming distance} between $x$ and $y$, namely, $\mathsf{dist}(x,y)=|\set{i \in [n]:x_i \ne y_i}|/n$.}\]
We say that $\Enc$ is $(\alpha,L)$-list-decodable if for every $w \in \B^n$, $|\List^{\Enc}_{\alpha}(w)| \le L$.
\end{definition}
The task of \emph{algorithmic} list-decoding is to produce the list $\List^{\Enc}_{\alpha}(w)$ on input $w \in \B^n$.

\emph{Local} unique decoding algorithms are algorithms that given an index $i \in [k]$, make few oracle queries to $w$, and reproduce the bit $m_i$ (with high probability over the choice of their random coins), where $m$ is such that
$\Enc(m)$ is the unique codeword close to $w$. This notion of \emph{local decoding} has many connections and applications in computer science and mathematics; see ~\cite{Yekhanin12} for a survey.

We will be interested in \emph{local} list-decoding algorithms
that receive oracle access to a received word $w \in \B^n$, as well as inputs $i \in [k]$ and $j \in [L]$. Informally, we will require that for every $m \in \List^{\Enc}_{\alpha}(w)$, with high probability, there exists a $j \in [L]$ such that for every $i \in [k]$, when the decoder $\Dec$ receives oracle access to $w$ and inputs $i,j$, it produces $m_i$ with high probability over its choice of random coins.
Local list decoding algorithms have many applications in theoretical computer science, for example in cryptography~\cite{GL89}, learning theory~\cite{KM93}, average-to-worst-case reductions \cite{Lip90}, and hardness amplification~\cite{BFNW93, STV01}. More formally, it is defined as follows.

\begin{definition}[Randomized local computation]
\label{dfn:randomized local computation}
We say that a procedure $P(i,R)$ \remph{locally computes} a string $m \in \B^k$ with error $\delta$, if for every $i \in [k]$, $\Pr[P(i,R)=m_i] \ge 1-\delta$ (where the probability is over a uniform choice of the ``string of random coins'' $R$).
\end{definition}

The definition of local list-decoders considers an algorithmic scenario that works in two steps:
\begin{itemize}
\item At the first step (which can be thought of as a preprocessing step) the local list-decoder $\Dec$ is given oracle access to $w$ and an index $j \in [L]$. It tosses random coins (which we denote by $r^{\shared}$).
\item At the second step, the decoder receives the additional index $i \in [k]$, and tosses additional coins $r$.
\item It is required that for every $w \in \B^n$ and $m \in \List^{\Enc}_{\alpha}(w)$, with probability at least $2/3$ over the choice of the shared coins $r^{\shared}$, there exists $j \in [L]$ such that when the local list-decoder receives $j$, it locally computes $m$ (using its ``non-shared'' coins $r$). The definition uses two types of random coins because the coins $r^{\shared}$ are ``shared'' between different choices of $i \in [k]$ and allow different $i$'s to ``coordinate''. The coins $r$, are chosen independently for different choices of $i \in [k]$.
\end{itemize}

This is formally stated in the next definition.

\begin{definition}[Local list-decoder]
\label{dfn:LLD}
Let $\Enc:\B^k \ar \B^n$ be a function. An \remph{$(\alpha,L,q,\delta)$-local list-decoder (LLD)} for $\Enc$ is an oracle procedure $\Dec^{(\cdot)}$ that receives oracle access to a word $w \in \B^n$, and makes at most $q$ calls to the oracle. The procedure $\Dec$ also receives inputs:
\begin{itemize}
\item $i \in [k]$ : The index of the symbol that it needs to decode.
\item $j \in [L]$ : An index to the list.
\item Two strings $r^{\shared},r$ that are used as random coins.
\end{itemize}
It is required that for every $w \in \B^n$, and for every $m \in \List^{\Enc}_{\alpha}(w)$, with probability at least $2/3$ over choosing a uniform string $r^{\shared}$, there exists $j \in [L]$ such that the procedure \[P_{w,j,r^{\shared}}(i,r)=\Dec^w(i,j,r^{\shared},r) \] locally computes $m$ with error $\delta$. If we omit $\delta$, then we mean $\delta=1/3$.
\end{definition}

\begin{remark}[On the generality of Definition \ref{dfn:LLD}]
The goal of this paper is to prove lower bounds on local list-decoders, and so, making local list-decoders as general as possible, makes our results stronger. We now comment on the generality of Definition \ref{dfn:LLD}.
\begin{itemize}
\item In Definition \ref{dfn:LLD} we do not require that $L = |\List^{\Enc}_{\alpha}(w)|$, and allow the local list-decoder to use a larger $L$. This means that on a given $w$, there may be many choices of $j \in [L]$ such that the procedure $P_{w,j,r^{\shared}}(i,r)=\Dec^w(i,j,r^{\shared},r)$ locally computes messages $m \not \in \List^{\Enc}_{\alpha}(w)$.
\item In Definition \ref{dfn:LLD} we do not place any restriction on the number of random coins used by the local list-decoder, making the task of local list-decoding easier.
\item We allow $\Dec$ to make \emph{adaptive} queries to its oracle.
\item We are only interested in the total number of queries made by the combination of the two steps. It should be noted that w.l.o.g., a local list-decoder can defer all its queries to the second step (namely, after it receives the input $i$), and so, this definition captures local list-decoding algorithms which make queries to the oracle at both steps.
\item To the best of our knowledge, all known local list-decoders in the literature are of the form presented in Definition \ref{dfn:LLD}.
\end{itemize}
\end{remark}

\subsubsection{Lower bounds on the query complexity of local list-decoders}
\label{sec:intro:code:results}

In this paper we prove lower bounds on the number of queries $q$ of $(\half-\eps,L,q,\delta)$-local list-decoders. Our goal is to show that the number of queries $q$ has to be large, when $\eps$ is \emph{small}.
Our lower bounds apply even if the size of the list $L$ is huge and approaches $2^k$ (note that a local list-decoder can trivially achieve $L=2^k$ with a list of all messages). Our lower bounds also apply even if the rate of the code is very small, and $n \ge 2^k$.

We remark that this parameter regime is very different than the one studied in lower bounds on the number of queries of local decoders for \emph{uniquely} decodable codes (that is, for $L=1$). By the Plotkin bound, uniquely decodable codes cannot have $\eps < \frac{1}{4}$, and so, the main focus in uniquely decodable codes is to show that local decoders for codes with ``good rate'' and ``large'' $\eps=\Omega(1)$, must make many queries. In contrast, we are interested in the case where $\eps$ is small, and want to prove lower bounds that apply to huge lists and small rate.

\paragraph{Lower bounds for large $\eps$.}
Our first result is a tight lower bound of $q=\Omega(\frac{\log(1/\delta)}{\eps^2})$ on the number of queries, assuming $\eps$ is sufficiently large, namely $\eps \ge \frac{1}{k^{\nu}}$ for some universal constant $0 < \nu \le 1$.

\begin{restatable}[Tight lower bounds for large $\eps$]{theorem}{largeeps}\label{thm:main:code:large eps}
There exists a universal constant $\nu >0$ such that for any $L \leq 2^{k^{0.9}}$, $\epsilon \in (k^{-\nu}, \frac 1 4)$,  and $\delta \in (k^{-\nu}, \frac 1 3)$, we have that
every $(\half-\eps,L,q,\delta)$-local list-decoder for $\Enc:\B^k \ar \B^n$ must have $q=\Omega(\frac{\log(1/\delta)}{\eps^2})$.
\end{restatable}

Theorem \ref{thm:main:code:large eps} is tight in the sense that the Hadamard code (which has length $n=2^k$) has $(\half-\eps,L=O(1/\eps^2),q=O(\frac{\log(1/\delta)}{\eps^2}),\delta)$ local list-decoders \cite{GL89}.
In fact, the Hadamard code was the motivation for this research, and is a running example in this paper.

Our results show that even if we allow list sizes $L$ which approach $2^k$, it is impossible to reduce the number of queries for the Hadamard code.
Our results also show that even if we are willing to allow very small rate ($n \geq 2^k$), and huge list sizes, it is impossible to have codes whose local list-decoders make fewer queries than the local list-decoders for the Hadamard code.

\paragraph{Comparison to previous work.}
Theorem \ref{thm:main:code:large eps} improves and extends an earlier work by Grinberg, Shaltiel and Viola \cite{GSV18} that gave the same bound of $q=\Omega(\frac{\log(1/\delta)}{\eps^2})$ for a more limited parameter regime: Specifically, in \cite{GSV18}, for the lower bound to hold, it is also required that $n \le 2^{k^{\gamma}}$, for some universal constant $\gamma>0$, and that the total number of coins tossed by the local list-decoder is less than $k^{\gamma}-\log L$ (which in particular implies that $L \leq 2^{k^\gamma}$).\footnote{The work of \cite{GSV18} is concerned with proving lower bounds on the number of queries of ``nonuniform reductions for hardness amplification'' \cite{ViolaThesis,SV08,AS11,GSV18}. As explained in \cite{ViolaThesis,SV08,AS11,GSV18} such lower bounds translate into lower bounds on local list-decoders, by ``trading'' the random coins of a local list-decoder for ``nonuniform advice'' for the reduction, and proving a lower bound on the number of queries made by the reduction.} We stress that because of these two limitations, the lower bounds of \cite{GSV18} do not apply to the Hadamard code and other low rate codes.

\paragraph{Extensions to large alphabet and erasures.}
The scenario that we consider in Theorem \ref{thm:main:code:large eps} has \emph{binary} alphabet, and decoding from \emph{errors}. We remark that in the case of large alphabets, or decoding from erasures, there are local list-decoders which achieve $q=O(\frac{\log(1/\delta)}{\eps})$ (which is smaller than what is possible for binary alphabet and decoding from errors), as was shown for the case of Hadamard codes and erasures in \cite{RRV18}, and for large alphabets in \cite{IJKW10}.
Our results can be extended to give a matching lower bound of $q=\Omega(\frac{\log(1/\delta)}{\eps})$ for decoding from erasures (for any alphabet size), and also the same lower bound on decoding from errors for any alphabet size.

\paragraph{Lower bounds for small $\eps$.}
The best bound on $q$ that Theorem \ref{thm:main:code:large eps} (as well as the aforementioned lower bounds of \cite{GSV18}) can give is $q \ge k^{\nu}$ for a universal constant $\nu >0$. The next theorem shows that even for small $\eps < 1/k$, we can obtain a lower bound on $q$ which is polynomial in $1/\eps$.

\begin{restatable}[Tight lower bounds for small $\eps$]{theorem}{smalleps}\label{thm:main:code:small eps}
There exist universal constants $\beta>0$ and $a,c>1$ such that for any $L \le \beta \cdot 2^{k}$,  $\eps \in ( \frac{a}{\sqrt{n}}, \frac 1 4)$, and $\delta <\frac{1}{3}$ and we have that
every $(\half-\eps,L,q,\delta)$-local list-decoder for $\Enc:\B^k \ar \B^n$ must have $q \ge \frac{1}{c \sqrt{\epsilon}\log k} -  \log L$.
\end{restatable}

Note that for sufficiently small $\eps = 1/(\log k)^{\omega(1)}$, we get $q=\Omega(\frac{1}{\eps^{1/2-o(1)}})$. It follows that together, Theorems \ref{thm:main:code:large eps} and Theorem \ref{thm:main:code:small eps} give a lower bound of $q=\Omega(\frac{1}{\eps^{1/2-o(1)}})$ that applies to every choice of $\eps \ge \Omega(\frac{1}{\sqrt{n}})$. To the best of our knowledge, Theorem \ref{thm:main:code:small eps} is the first lower bound on local list-decoders that is able to prove a lower bound of $q \ge k$ (and note that this is what we should expect when $\eps < \frac{1}{k})$.
We also remark that the requirement that $\eps$ is not too small compared to $n$ (as is made in Theorem \ref{thm:main:code:small eps}) is required (as we cannot prove lower bounds on the number of queries in case $\eps<\frac{1}{n}$).

Goldreich and Levin \cite{GL89} showed that locally list-decodable codes with small $\eps < 1/k$ can be used to give constructions of hard-core predicates. We explain this connection in the next section.

\subsection{Hard-core predicates}

The celebrated Goldreich-Levin theorem \cite{GL89} considers the following scenario: There is a computational task where the required output is \emph{non-Boolean} and is hard to compute on average. We would like to obtain a \emph{hard-core predicate}, which is a \emph{Boolean} value that is hard to compute on average.

The Goldreich-Levin theorem gives a solution to this problem, and in retrospect, the theorem can also be viewed as a $(\half-\eps,L^{\Hadamard}=O(\frac{1}{\eps^2}),q^{\Hadamard}=O(\frac{k}{\eps^2}),\delta=\frac{1}{2k})$-local list-decoder for the Hadamard code, defined by: $\Enc^{\Hadamard}:\B^k \ar \B^{n=2^k}$, where for every $r \in \B^k$, \[ \Enc^{\Hadamard}(x)_r=\left(\sum_{i \in [k]}x_i \cdot r_i \right)\mod 2. \]
In retrospect, the Goldreich-Levin theorem can also be seen as showing that \emph{any} locally list-decodable code with suitable parameters can be used to produce hard-core predicates.

We consider two scenarios for the Goldreich-Levin theorem depending on whether we want to extract a hard-core bit from a function $g:\B^{\ell} \ar \B^{\ell}$ that is \emph{hard to compute} on a random input, or to extract a hard-core bit from a one-way function $f:\B^{\ell} \ar \B^{\ell}$ that is \emph{hard to invert} on a random output.

\subsubsection{Functions that are hard to compute}
\label{sec:intro:HC:compute}

Here the goal is to transform a \emph{non-Boolean function} $g$ that is hard to compute on a random input, into a \emph{predicate} $g^{\pred}$ that is hard to compute on a random input. More precisely:

\begin{itemize}
\item\textit{Assumption:} There is a non-Boolean function that is hard to compute with probability $\rho$.

    \smallskip
    Namely, a function $g:\B^{\ell} \ar \B^{\ell}$ such that for every circuit $C$ of size $s$, \[\Pr_{x \from U_{\ell}}[C(x)=g(x)] \le \rho.\footnote{We use $U_{\ell}$ to denote the uniform distribution on $\ell$ bits.} \]
\item \textit{Conclusion:} There is a predicate $g^{\pred}:\B^{\ell'} \ar \B$ that is hard to compute with probability $\half+\eps$.
    Namely, for every circuit $C'$ of size $s'$, \[\Pr_{x \from U_{\ell'}}[C'(x)=g^{\pred}(x)] \le \half+\eps. \]
\item \textit{Requirements:} The goal is to show that for every $g$, there exists a function $g^{\pred}$ with as small an $\eps$ as possible, without significant losses in the other parameters (meaning that $s'$ is not much smaller than $s$, and $\ell'$ is not much larger than $\ell$).
\end{itemize}

The Goldreich-Levin theorem for this setting can be expressed as follows.

\begin{theorem}[Goldreich-Levin for functions that are hard to compute~\cite{GL89}]
\label{thm:GL compute}
For a function $g:\B^{\ell} \ar \B^{\ell}$, define $g^{\pred}:\B^{\ell'=2\ell} \ar \B$ by $g^{\pred}(x,r)=\Enc^{\Hadamard}(g(x))_r$, and  $\rho = \frac{\eps}{2 \cdot L^{\Hadamard}}=\poly(\eps)$. If
for every circuit $C$ of size $s$,
\[ \Pr_{x \from U_\ell}[C(x)=g(x)] \le \rho, \]
then for every circuit $C'$ of size $s'=\frac{s}{q^{\Hadamard} \cdot \poly(\ell)} = s \cdot \poly(\frac{\eps}{\ell})$,
\[ \Pr_{x \from U_{2\ell}}[C'(x)=g^{\pred}(x)] \le \half+\eps. \]
\end{theorem}

The Hadamard code can be replaced by any locally list-decodable code with list size $L$ for decoding from radius $\half-\eps$, with $q$ queries for $\delta = 1/(2k)$. For such a code (assuming also that the local list-decoder can be computed efficiently) one gets the same behavior. Specifically, if the initial function is sufficiently hard and $\rho=\frac{\eps}{2L}$, then the Boolean target function is hard to compute, up to $\half+\eps$ for circuits of size roughly $s'=s/q$.

\paragraph{Is it possible to improve the Goldreich-Levin theorem for $\rho \ll 1/s$?}
Suppose that we are given a function $g:\B^{\ell} \ar \B^{\ell}$ that is hard to compute for circuits of size $s=\poly(\ell)$, with success, say, $\rho=1/2^{\sqrt{\ell}}$. When applying Theorem \ref{thm:GL compute}, we gain nothing compared to the case that $\rho=1/\poly(\ell)$. In both cases, we can obtain $\eps=1/\poly(\ell)$, but not smaller! (Since otherwise $s'= s \cdot \poly(\epsilon/\ell)$ is smaller than $1$ and the result is meaningless).

This is disappointing, as we may have expected to obtain $\eps \approx \rho = 1/2^{\sqrt{\ell}}$, or at least, to gain over the much weaker assumption that $\rho=1/\poly(\ell)$. This leads to the following open problem:

\begin{oprob}[Improve Goldreich-Levin for functions that are hard to compute]
\label{oprob:compute}
Suppose we are given a function $g:\B^{\ell} \ar \B^{\ell}$ such that circuits $C$ of size $s=\poly(\ell)$ cannot compute $g$ with success $\rho = 1/2^{\sqrt{\ell}}$. Is it possible to convert $g$ into a predicate with hardness $\half+\eps$ for $\eps=1/\ell^{\omega(1)}$?
\end{oprob}

This is not possible to achieve using the Hadamard code, because the number of queries is $q \ge 1/\eps$, and Theorem \ref{thm:GL compute} requires $s \ge s' \cdot q \ge q \ge 1/\eps$, which dictates that $\eps \ge 1/s$.

Note that when $\rho$ is small, we can afford list-decodable codes with huge list sizes of $L \approx 1/\rho$. Motivated by this observation, we can ask the following question:

\begin{quote}
Is it possible to solve  Open Problem \ref{oprob:compute} by substituting the Hadamard code with a better code? Specifically, is it possible for local list-decoders to have $q = \frac{1}{\eps^{o(1)}}$ if allowed to use \emph{huge} lists of size say $2^{\sqrt{k}}$, approaching the trivial bound $2^k$? (Note that in the Hadamard code, the list size used is $\poly(1/\eps)=\poly(k)$ which is exponentially smaller).
\end{quote}

Theorem \ref{thm:main:code:small eps} shows
that it is impossible to solve Open Problem  \ref{oprob:compute} by replacing the Hadamard code with a different locally list-decodable code.

The natural next question is whether we can use other techniques (not necessarily local list-decoding) to achieve the goal stated above. In this paper, we show that this cannot be done by \emph{black-box techniques}:

\begin{informal theorem}[Black-box impossibility result for functions that are hard to compute]
\label{ithm:compute}
If $\rho \ge 2^{-\ell/5}$ and $\eps=\frac{1}{s^{\omega(1)}}$,
then there does not exist a map
that converts a function $g$ into a function $g^{\pred}$ together with a black-box reduction showing that $g^{\pred}$ is a hard-core predicate for $g$.
\end{informal theorem}

The parameters achieved in Theorem \ref{ithm:compute} rule out black-box proofs in which $\eps=\frac{1}{s^{\omega(1)}}$, not only for $s=\poly(\ell)$ and $\rho=2^{-\sqrt{\ell}}$ (as in Open Problem \ref{oprob:compute}) but also for $\rho=2^{-\Omega(\ell)}$, and allowing  $s$ as large as  $2^{\ell}$.

The precise statement of Theorem \ref{ithm:compute} is stated in Theorem \ref{thm:hard-core compute}, and the precise model is explained in Section \ref{sec:HC compute}.
To the best of our knowledge, this is the first result of this kind, that shows black-box impossibility results for Open Problem \ref{oprob:compute}. Moreover, we believe that the model that we introduce in Section \ref{sec:HC compute} is very general and captures all known black-box techniques. In particular, our model  allows the reduction to introduce nonuniformity when converting an adversary $C'$ that breaks $g^{\pred}$ into an adversary $C$ that breaks $g$. See discussion in Remark \ref{rem:model}.

\subsubsection{Functions that are hard to invert}
\label{sec:intro:HC:invert}

Here the goal is to transform a one-way function $f$ into a new one-way function $f^{\newowf}$ and a \emph{predicate} $f^{\pred}$ such that it is hard to compute $f^{\pred}(x)$ given $f^{\newowf}(x)$. More precisely:
\begin{itemize}
\item \textit{Assumption:} There is a one-way function that is hard to invert with probability $\rho$.

    \smallskip
    Namely, a function $f:\B^{\ell} \ar \B^{\ell}$ such that for every circuit $C$ of size $s$, \[ \Pr_{x \from U_{\ell}}[C(f(x)) \in f^{-1}(f(x))] \le \rho. \]
\item \textit{Conclusion:} There is a one-way function $f^{\newowf}:\B^{\ell'} \ar \B^{\ell'}$, and a predicate  $f^{\pred}:\B^{\ell'} \ar \B$, such that it is hard to predict $f^{\pred}(x)$ with advantage $\half+\eps$, when given access to $f^{\newowf}(x)$.
    Namely, for every circuit $C'$ of size $s'$, \[ \Pr_{x \from U_{\ell'}}[C'(f^{\newowf}(x)) = f^{\pred}(x)] \le \half+\eps. \]
\item \textit{Requirements:}  The goal is to show that for every $f$, there exist functions $f^{\newowf},f^{\pred}$ with as small an $\eps$ as possible, without significant losses in the other parameters (meaning that $s'$ is not much smaller than $s$, and $\ell'$ is not much larger than $\ell$).
\end{itemize}

The Goldreich-Levin theorem for this setting can be expressed as follows.

\begin{theorem}[Goldreich-Levin for functions that are hard to invert \cite{GL89}]
\label{thm:GL invert}
For a function $f:\B^{\ell} \ar \B^{\ell}$, define $f^{\newowf}:\B^{2\ell} \ar \B^{2\ell}$ by $f^{\newowf}(x,r)=(f(x),r)$, $f^{\pred}:\B^{2\ell} \ar \B$ by $f^{\pred}(x,r)=\Enc^{\Hadamard}(x)_r$, and $\rho = \frac{\eps}{2 \cdot L^{\Hadamard}}=\poly(\eps)$. If
for every circuit $C$ of size $s$,
\[ \Pr_{x \from U_\ell}[C(f(x)) \in f^{-1}(f(x))] \le \rho, \]
then for every circuit $C'$ of size $s'=\frac{s}{q^{\Hadamard} \cdot \poly(\ell)} = s \cdot \poly(\frac{\eps}{\ell})$,
\[ \Pr_{x \from U_{2\ell}}[C'(f^{\newowf}(x))=f^{\pred}(x)] \le \half+\eps. \]
\end{theorem}

\begin{remark}
The problem of obtaining a hard-core predicate for one-way functions is interesting only if an unbounded adversary $\phi:\B^{\ell'} \ar \B$ can predict $f^{\pred}(x)$ when given $f^{\newowf}(x)$ as input. If this is not required, then one can take $\ell'=\ell+1$, $f^{\pred}(x)=x_1$, and $f^{\newowf}(x_1,\ldots,x_{n+1})=f(x_2,\ldots,x_{n+1})$. However, this is trivial, and is not useful in applications. Therefore, when considering this problem, we will assume that there exists such a $\phi:\B^{\ell'} \ar \B$.

A natural example is the case where the original one-way function $f$ and the constructed function $f^{\newowf}$ are one-way permutations.
In fact, in the case that $f,f^{\newowf}(x)$ are permutations, the setup of ``functions that are hard to invert'' can be seen as a special case of the setup of ``functions that are hard to compute'' by taking $g=f^{-1}$, and $g^{\pred}(y)=f^{\pred}((f^{\newowf})^{-1}(y))$.

We point out that, in this setting, the circuit $C'$ that is trying to invert $f$ (that is, to compute $g$) has an advantage over its counterpart in the setup of ``functions that are hard to compute'': It can use the efficient algorithm that computes the ``forward direction'' of $f$, when trying to invert $f$. In terms of $g$, this means that the circuit $C'$ can compute $g^{-1}$ for free. This distinction is explained in Section \ref{sec:HC invert}.
\end{remark}

\paragraph{Is it possible to improve the Goldreich-Levin theorem for $\rho \ll 1/s$?}
The same problem that we saw with functions that are hard to compute, also shows up in the setup of functions that are hard to invert.
Suppose that we are given a function $f:\B^{\ell} \ar \B^{\ell}$ that is hard to invert for circuits of size $s=\poly(\ell)$ with success, say, $\rho=1/2^{\sqrt{\ell}}$. When applying Theorem \ref{thm:GL invert}, we gain nothing compared to the case that $\rho=1/\poly(\ell)$. In both cases, we can obtain $\eps=1/\poly(\ell)$, but not smaller! This is expressed in the next open problem:

\begin{oprob}[Improve Goldreich-Levin for functions that are hard to invert]
\label{oprob:invert}
Suppose we are given a one-way function $f:\B^{\ell} \ar \B^{\ell}$ such that circuits $C$ of size $s=\poly(\ell)$ cannot invert $f$ with success $\rho = 1/2^{\sqrt{\ell}}$. Is it possible to obtain a hard-core predicate $f^{\pred}$ with hardness $\half+\eps$ for $\eps=1/\ell^{\omega(1)}$ for some choice of one-way function $f^{\newowf}$?
\end{oprob}

In this paper, we show that this cannot be done by \emph{black-box techniques}. The formulation of Theorem \ref{ithm:invert} below, is very similar to that of Theorem \ref{ithm:compute} with some small modification in the parameters.

\begin{informal theorem}[Black-box impossibility result for functions that are hard to invert]
\label{ithm:invert}
If $\rho \ge 2^{-\ell/5}$, $s=2^{o(\ell)}$, and $\eps=\frac{1}{s^{\omega(1)}}$, then there does not exist
a map that converts a function $f$ into functions $f^{\newowf},f^{\pred}$ together with a black-box reduction showing that $f^{\pred}$ is a hard-core predicate for $f^{\newowf}$.
\end{informal theorem}

The precise statement of Theorem \ref{ithm:invert} is stated in Theorem \ref{thm:hard-core invert}, and the precise model is explained in Section \ref{sec:HC invert}.
To the best of our knowledge, this is the first result of this kind, that shows black-box impossibility results for Open Problem \ref{oprob:invert}. Moreover,
we believe that the model that we introduce in Section \ref{sec:HC invert} is very general, and captures all known black-box techniques. In particular, our model allows the reduction to compute the easy direction of the function $f$, and to introduce nonuniformity when converting an adversary $C'$ that breaks $f^{\pred}$ into an adversary $C$ that breaks $f$.

\begin{remark}[The model we use for black-box proofs]
\label{rem:model}
Many different models of ``black-box techniques'' for cryptographic primitives were studied in the literature and the reader is referred to \cite{RTV04} for a discussion and a taxonomy.
The model for ``black-box techniques''  that we use is described in detail in Section \ref{sec:Hard-core}. The notion that we use is more general than the notion of ``fully-BB reductions'' discussed in \cite{RTV04}, and incomparable to the other notions discussed in \cite{RTV04}, specifically:
\begin{itemize}
\item We require that there is a ``transformation map'' which given any function $f$ produces functions $f^{\newowf}$ and $f^{\pred}$, however, unlike the notions studied in \cite{RTV04}, we do not make the requirement that this transformation map can be efficiently computed.
\item We require that there is a reduction $\Red$ such that for any $f$, and for every adversary $C'$ (not necessarily efficient) breaking the security of $f^{\pred}$, $\Red^{f,C'}$ can be used to invert $f$. So far, this is similar in spirit to the notion of ``fully-BB reduction'' defined in \cite{RTV04}. However, we give reductions more power than is given in \cite{RTV04} in the case of fully-BB reductions. Specifically, we also allow $\Red$ to introduce nonuniformity (that could depend on $C'$ and $f$). More formally, for every adversary $C'$ that breaks the security of $f^{\pred}$, we require that there exists a short nonuniform advice string $\alpha$ such that $\Red^{C',f}(\cdot,\alpha)$ inverts $f$.
\end{itemize}
\end{remark}

\subsection{More related work}

\paragraph{Lower bounds on the number of queries of local decoders for \emph{uniquely} decodable codes.}
In this paper, we prove lower bounds on the number of queries of local list-decoders.
There is a long line of work that is concerned with proving lower bounds on the number of queries of uniquely decodable codes. As we have explained in Section \ref{sec:intro:code:results}, the parameter regime considered in the setting of uniquely decodable codes is very different than the parameter regime we consider here \cite{Yekhanin12}.

\paragraph{Lower bounds on nonuniform black-box reductions for hardness amplification.}
A problem that is closely related to proving lower bounds on the number of queries of local list-decoders is the problem of proving lower bounds on the number of queries of nonuniform black-box reductions for hardness amplification. We have already discussed this line of work \cite{ViolaThesis,SV08,AS11,GSV18,S20} in Section \ref{sec:intro:code:results}.

Lower bounds on such reductions can be translated to lower bounds on local list-decoders (as long as the number of coins tossed by the local list-decoders is small). We remark that for the purpose of hardness amplification, it does not make sense to take codes with small rate (namely, codes with $n = 2^{k^{\Omega(1)}}$). The focus of Theorem \ref{thm:main:code:large eps} is to handle such codes.
Additionally, when using codes for hardness amplification, it does not make sense to take $\eps<1/k$ (or even $\eps<1/\sqrt{k}$). In contrast, the parameter regime considered in Theorem \ref{thm:main:code:small eps} focuses on small $\eps$.

Motivated by hardness amplification, there is also a related line of work studying limitations on the complexity of local list decoders (and specifically, whether these decoders need to compute the majority function) \cite{ViolaThesis,SV08,GR08,AS11,GSV18,S20}.
Another approach to prove limitations on hardness amplification is to show that assuming certain cryptographic assumptions, hardness amplification that is significantly better than what is currently known is impossible, see e.g., \cite{DJMW12} for a discussion.

\paragraph{Other improvements of the Goldreich-Levin theorem.}
In this paper, we are interested in whether the Goldreich-Levin theorem can be improved. Specifically, we are interested in improvements where, when the original function has hardness $\rho=2^{-\Omega(\ell)}$ for polynomial size circuits, then the hard-core predicate has hardness $\half+\eps$ for $\eps=\ell^{-\omega(1)}$.
We remark that there are other aspects of the Goldreich-Levin theorem that one may want to improve.
\begin{itemize}
\item When given an initial non-Boolean function on $\ell$ bits, the Goldreich-Levin theorem produces a hard-core predicate on $\ell'=2\ell$ bits. It is possible to make $\ell'$ smaller (specifically, $\ell'=\ell+O(\log(1/\eps))$ by using other locally list-decodable codes instead of Hadamard. Our limitations apply to \emph{any} construction (even one that is not based on codes) and in particular also for such improvements.
\item It is sometimes desirable to produce many hard-core bits (instead of the single hard-core bit that is obtained by a hard-core predicate). This can be achieved by using ``extractor codes'' with a suitable local list-decoding algorithm. The reader is referred to \cite{TaShmaZuckerman} for more details. Once again, our limitations obviously apply also for the case of producing many hard-core bits.
\end{itemize}

\subsection*{Open problems and subsequent work}

We end this section with a couple of interesting open problems for future research.

\begin{enumerate}
\item Unlike Theorem \ref{thm:main:code:large eps} (that handles large $\eps$), Theorem \ref{thm:main:code:small eps} (that handles small $\eps$) does not achieve a bound of $q=\Omega(\frac{\log(1/\delta)}{\eps^2})$, and only achieves a bound of $\Omega(\frac{1}{\sqrt{\eps}})$. A natural open problem is to improve the bound on $q$ for small $\eps$ to match the bound for large $\eps$.

    A step in this direction was made in the subsequent work of Shaltiel and Viola \cite{SV22}. They show that if $L \le 2^{\frac{k}{20}}$, $\delta < \frac{1}{3}$ and $n=\Omega(\frac{1}{\eps^2})$ (which are slightly stronger conditions than the one we require in Theorem \ref{thm:main:code:small eps}) then the bound on the number of queries can be improved from $q=\Omega(\frac{1}{\sqrt{\eps} \cdot \log k})-O(\log L)$ to $q=\Omega(\frac{1}{\eps \cdot \log k})$.

    The improvement of \cite{SV22} builds on the methodology used in this paper (namely, converting a local list decoder into a depth 3 circuit, and then employing a bound on depth 3 circuits to obtain a lower bound on the number of queries). However, rather than reducing to a depth 3 circuit for the coin problem (as we do in this paper) Shaltiel and Viola \cite{SV22} note that the reduction used in this paper gives a depth 3 circuit for a more specific problem. They then show a quantitatively stronger lower bound for this specific problem, and this results in the improved bound on the number of queries.

\item In the case of large $\eps$, Theorem \ref{thm:main:code:large eps} can be extended to handle local list-decoding from \emph{erasures}, and gives a lower bound of $q=\Omega(\frac{\log(1/\delta)}{\eps})$ on the number of queries of local list-decoders that decode from a $1-\eps$ fraction of erasures. We do not see how to extend the proof of Theorem \ref{thm:main:code:small eps} to erasures.

\item The model of black-box proofs that we introduce in Section \ref{sec:Hard-core} is quite general, and to the best of our knowledge, covers all known proofs in the literature on hard-core predicates for general one-way functions. Is it possible to circumvent the black-box limitations and answer open problems \ref{oprob:compute} and \ref{oprob:invert} for \emph{specific candidates} for one-way functions?
More generally, is it possible to come up with non-black-box techniques that circumvent the limitations?
\end{enumerate}

\subsection*{Organization of the paper}

We give a high level overview of our techniques in Section \ref{sec:technique}.
Our results on local list-decoders (and the proofs of Theorem \ref{thm:main:code:large eps} and Theorem \ref{thm:main:code:small eps}) are presented in Section \ref{sec:LLDC}. Our results on hard-core predicates appear in Section \ref{sec:Hard-core} (which includes a precise description of the model and formal restatements of Theorem \ref{ithm:compute} and
Theorem \ref{ithm:invert}).

\section{Techniques}
\label{sec:technique}

In this section we give a high level overview of our techniques.
Our approach builds on earlier work for proving lower bounds on the number of queries of reductions for hardness amplification \cite{ViolaThesis,SV08,GSV18}. In this section, we give a high level overview of the arguments used to prove our main theorems.

\subsection{Local list-decoders on random noisy codewords}

Following \cite{ViolaThesis,SV08,GSV18}, we will consider a scenario which we refer to as ``random noisy codewords'' in which a uniformly chosen message $m$ is encoded, and the encoding is corrupted by a binary symmetric channel.

\begin{definition}[Binary symmetric channels]
Let $\BSC^n_p$ be the experiment in which a string $Z \in \B^n$ is sampled, where $Z=Z_1,\ldots,Z_n$ is composed of i.i.d. bits, such that for every $i \in [n]$, $\Pr[Z_i=1]=p$.
\end{definition}

\begin{definition}[Random noisy codewords]
\label{dfn:RNSY}
Given a function $\Enc:\B^k \ar \B^n$ and $p>0$ we consider the following experiment (which we denote by $\RNSY^{\Enc}_p$):
\begin{itemize}
\item A message $m \from \B^k$ is chosen uniformly.
\item A noise string $z \from \BSC^n_p$ is chosen from a binary symmetric channel.
\item We define $w=\Enc(m) \oplus z$.
\end{itemize}
We use $(m,z,w) \from \RNSY^{\Enc}_p$ to denote $m,z,w$ which are sampled by this experiment. We omit $\Enc$ if it is clear from the context.
\end{definition}

Our goal is to prove lower bounds on the number of queries $q$ of a $(\half-\eps,L,q,\delta)$-local list-decoder $\Dec$ for a code $\Enc:\B^k \ar \B^n$. For this purpose, we will consider the experiment $\RNSY_p$ for the values $p=\half-2\eps$ and $p=\half$.

For $p=\half-2\eps$, and $(m,z,w) \from \RNSY_{\half-\eps}$, by a Chernoff bound, the relative Hamming weight of $z$ is, with very high probability, less than $\half-\eps$. This implies that $\mathsf{dist}(w,\Enc(m)) \le \half-\eps$, meaning that $m \in \List_{\half-\eps}^{\Enc}(w)$. It follows that there must exist $j \in [L]$ such that when given input $j$, and oracle access to $w$, the decoder $\Dec$ recovers the message $m$.

For $p=\half$, and $(m,z,w) \from \RNSY_{\half}$, the string $z$ is uniformly distributed and independent of $m$. This means that $w=\Enc(m) \oplus z$ is uniformly distributed and independent of $m$. Consequently, when $\Dec$ is given oracle access to $w$, there is no information in $w$ about the message $m$, and so, for every $j \in [L]$, the probability that $\Dec$ recovers $m$ when given input $j$ and oracle access to $w$ is exponentially small.

Loosely speaking, this means that $\Dec$ can be used to distinguish $\BSC^n_{\half-2\eps}$ from $\BSC^n_{\half}$. It is known that distinguishing these two distributions requires many queries. We state this informally below, and a formal statement appears in Lemma \ref{lem:q queries}.

\begin{informal theorem}
\label{ithm:q queries}
Any function $T:\B^q \ar \B$ that distinguishes $\BSC^q_{\half-2\eps}$ from $\BSC^q_{\half}$ with advantage\footnote{We say that a function $T:\B^q \ar \B$ distinguishes a pair of distributions $B_1$, $B_2$ over $\{0,1\}^q$ with \emph{advantage} $\delta$ if $\Pr_{x_1 \leftarrow B_1, x_2 \leftarrow B_2}\left[ T(x_1)=T(x_2) \right] \leq \frac 1 2 +\delta$.} $\delta$, must have $q=\Omega(\frac{\log(1/\delta)}{\eps^2})$.
\end{informal theorem}

Thus, in order to prove a tight lower bound of $q=\Omega(\frac{\log(1/\delta)}{\eps^2})$, it is sufficient to show how to convert a $(\half-\eps,L,q,\delta)$-local list-decoder $\Dec$, into a function $T$ that distinguishes $\BSC^q_{\half-2\eps}$ from $\BSC^q_{\half}$ with advantage $\delta$. Note that we can allow $T$ to be a ``randomized procedure'' that tosses coins, as by an averaging argument, such a randomized procedure can be turned into a deterministic procedure.

\subsection{Warmup: the case of unique decoding}
\label{sec:intro:technique:warmup}

Let us consider the case that $L=1$ (that is unique decoding). We stress that this case is uninteresting, as by the Plotkin bound, it is impossible for nontrivial codes to be uniquely decodable for $\eps<\frac{1}{4}$, and so, there are no local decoders for $L=1$ and $\epsilon < \frac 1 4$, regardless of the number of queries. Nevertheless, this case will serve as a warmup for the approach we use later.

Our goal is to convert $\Dec$ into a randomized procedure $T:\B^q \ar \B$ that distinguishes $\BSC^q_{\half-2\eps}$ from $\BSC^q_{\half}$.
The procedure $T$ will work as follows: On input $z \from \B^q$, we choose $m \from \B^k$, and $i \from [k]$. We then run $\Dec$ on input $i$, and when $\Dec$ makes its $t$'th query $\ell_t \in [n]$ to the oracle, we answer it by $\Enc(m)_{\ell_t} \oplus z_t$. That is, we answer as if $\Dec$ is run with input $i$ and oracle access to $w=\Enc(m) \oplus z$, for $z$ chosen from a binary symmetric channel. The final output of $T$ is whether $\Dec$ reproduced $m_i$. This procedure $T$ simulates $\Dec^{w}(i)$, and therefore distinguishes $\BSC^q_{\half-2\eps}$ from $\BSC^q_{\half}$, implying the desired lower bound.

Both Theorem \ref{thm:main:code:large eps} and Theorem \ref{thm:main:code:small eps} will follow by modifying the basic approach to handle $L>1$. In the remainder of this section, we give a high level overview of the methods that we use. The formal section of this paper does not build on this high level overview, and readers can skip this high level overview and go directly to the formal section if they wish to.

\subsection{Reducing to the coin problem for AC$^0$}

We start with explaining the approach of proving Theorem \ref{thm:main:code:small eps}.
Consider a randomized procedure $C$ that on input $z \in \B^n$, chooses $m \from \B^k$ and prepares $w=\Enc(m) \oplus z$. The procedure then computes $\Dec^w(i,j)$ for all choices of $i \in [k]$ and $j \in [L]$ and accepts if there exists a $j \in [L]$ such that $\Dec^w(\cdot, j)$ recovers $m$. By the same rationale as in Section \ref{sec:intro:technique:warmup}, $C$ distinguishes $\BSC^n_{\half-2\eps}$ from $\BSC^n_{\half}$. This does not seem helpful, because $C$ receives $n$ input bits, and we cannot use Theorem \ref{ithm:q queries} to get a lower bound on $q$.

Inspired by a lower bound on the size of nondeterministic reductions for hardness amplification due to Applebaum et al. \cite{ASSY}, we make the following observation:
The procedure $C$ can be seen as $k \cdot L$ computations (one for each choice of $i \in [k]$ and $j \in [L]$) such that:
\begin{itemize}
\item These $k \cdot L$ computations can be run \emph{in parallel}.
\item Once these computations are made, the final answer $C(z)$ is computed by a constant-depth circuit.
\item Each of the $k \cdot L$ computations makes $q$ queries into $z$, and therefore can be simulated by a size $O(q \cdot 2^q)$ circuit of depth 2.
\end{itemize}
Overall, this means that we can implement $C$ by a circuit of size $s=\poly(k,L,2^q)$ and constant depth. (In fact, a careful implementation gives depth 3).

This is useful because there are lower bounds showing that small constant-depth circuits cannot solve the ``coin problem''. Specifically, by the results of Cohen, Ganor and Raz \cite{CGR14} circuits of size $s$ and depth $d$ cannot distinguish $\BSC^n_{\half-2\eps}$ from $\BSC^n_{\half}$ with constant advantage, unless $s \ge \mbox{exp}(\Omega(\frac{1}{\eps^{d-1}}))$.\footnote{These results of \cite{CGR14} improve upon earlier work of Shaltiel and Viola \cite{SV08} that gave slightly worse bound. These results are tight as shown by Limaye et al. \cite{LSSTV19} (that also extended the lower bound to hold for circuits that are also allowed to use parity gates).} This gives the bound stated in Theorem \ref{thm:main:code:small eps}.

We find it surprising that an \emph{information theoretic} lower bound on the number of queries of local list-decoders is proven by considering concepts like constant-depth circuits from \emph{circuit complexity}.

\paragraph{Extending the argument to lower bounds on hard-core predicates.}
It turns out that this argument is quite versatile, and this is the approach that we use to prove Theorems \ref{ithm:compute} and \ref{ithm:invert}. Loosely speaking, in these theorems, we want to prove a lower bound on the number of queries made by a reduction that, when receiving oracle access to an adversary that breaks the hard-core predicate, is able to compute (or invert) the original function too well. Such lower bounds imply that such reductions do not produce small circuits when used in black-box proofs for hard-core predicates.

We will prove such lower bounds by showing that a reduction that makes $q$ queries can be used to construct a circuit of size $s \approx 2^q$ and constant depth that solves the coin problem. Interestingly, this argument crucially relies on the fact that constant-depth circuits \emph{can} distinguish $\BSC^n_{\eps}$ from $\BSC^n_{2\eps}$ with size $\poly(n/\eps)$ which follows from the classical results of Ajtai on constant depth circuits for approximate majority \cite{Ajtai}.\footnote{The proof of Theorem \ref{ithm:invert} uses an additional versatility of the argument (which we express in the terminology of codes): The argument works even if the individual procedures that are run in parallel are allowed to have some limited access to the message $m$, as long as this does not enable them to recover $m$. This property is used to handle reductions in a cryptographic setup, where reductions have access to the easy direction of a one-way function.}

\subsection{Conditioning on a good $j$}

A disadvantage of the approach based on the coin problem is that at best, it can give lower bounds of $q =\Omega(1/\sqrt{\eps})$, and cannot give tight lower bounds of the form $q=\Omega(\frac{\log(1/\delta)}{\eps^2})$. In order to achieve such a bound (as is the case in Theorem \ref{thm:main:code:large eps}) we will try to reduce to Theorem \ref{ithm:q queries} which \emph{does} give a tight bound (in case $\eps$ is not too small).

Our approach builds on the earlier work of Grinberg, Shaltiel and Viola \cite{GSV18} that we surveyed in Section \ref{sec:intro:code:results}.
When given a $(\half-\eps,L,q,\delta)$-local list-decoder $\Dec$, we say that an index $j \in [L]$ is \emph{decoding} for $m,w$, if  when $\Dec$ is given oracle access to $w$ and input $j$, then with probability at least $1-10\delta$ over $i \in [k]$, we have that $\Dec^w(i,j)$ recovers $m_i$.

We use a careful averaging argument to show that there exists an index $j' \in [L]$, and a fixed choice of the random coins of $\Dec$, such that in the experiment $(m,z,w) \from \RNSY_{\half-2\eps}$,
$j'$ is decoding for $m,w$ with probability at least $\Omega(1/L)$. We then consider the experiment $\RNSY'_{\half-2\eps}$ in which we choose $(m,z,w) \from \RNSY_{\half-2\eps}$ \emph{conditioned} on the event $\set{\mbox{$j'$ is decoding for $m,w$}}$.

We have made progress, because in the experiment $\RNSY'_{\half-2\eps}$ there exists a unique $j'$ that is decoding, and so, when we implement the strategy explained in Section \ref{sec:intro:technique:warmup} we only need to consider this \emph{single} $j'$, which intuitively means that our scenario is similar to the warmup scenario of unique decoding described in Section \ref{sec:intro:technique:warmup}.

The catch is that when choosing $(m,z,w) \from \RNSY'_{\half-2\eps}$, we no longer have that $z$ is distributed like $\BSC^n_{\half-2\eps}$ (as the distribution of $z$ may be skewed by conditioning on the event that $j'$ is decoding).

Shaltiel and Viola \cite{SV08} (and later work \cite{GSV18,S20}) developed tools to handle this scenario. Loosely speaking, using these tools, it is possible to show that a large number of messages $m$ are ``useful'' in the sense that there exists an event $A_m$ such that if we consider $(m,z,w)$ that are chosen from $\RNSY'_{\half-2\eps}$ \emph{conditioned} on $A_m$, then there exists a subset $B(m) \subseteq [n]$ of small size $b$, such that $z|_{B(m)}$ is fixed, and $z|_{[n] \setminus B(m)}$ is distributed like $\BSC^{n-b}_{\half-2\eps}$.

If the number of possible choices for sets $B(m)$ is small, then by the pigeon-hole principle, there exists a fixed choice $B$ that is good for a large number of useful messages $m$. This can be used to imitate the argument we used in the warmup, and prove a lower bound.\footnote{Loosely speaking, this is because for useful messages, in the conditioned experiment, $z$ is distributed like $\BSC_{\half-2\eps}$ (except that some bits of $z$ are fixed as a function of $m$). Furthermore, as there are many useful messages, the local list-decoder does not have enough information to correctly recover the message when given oracle access to $\Enc(m) \oplus \BSC^n_{\half}=\BSC^n_{\half}$.}

\paragraph{Extending the argument to the case of small rate.}
A difficulty, that prevented \cite{GSV18} from allowing length as large as $n=2^k$,
is that $B(m)$ is a subset of $[n]$, and so, even if $b=|B(m)|=1$, the number of possible choices for such sets is at least $n$. For the pigeon-hole principle argument above, we need that the number of messages (that is $2^k$) is much larger than the number of possible choices for $B(m)$ (which is at least $n$). This means that one can only handle $n$ which is sufficiently smaller than $2^k$, and this approach cannot apply to codes with small rate (such as the Hadamard code).

We show how to solve this problem, and prove lower bounds for small rate codes. From a high level, our approach can be explained as follows: We consider the distribution of $B(m)=\set{Y_1(m) < \ldots <Y_b(m)}$ for a uniformly chosen useful $m$.
We first show that if all the $Y_j$'s have large min-entropy, then it is possible to prove a lower bound on $q$ by reducing to Theorem \ref{ithm:q queries} (the details of this are explained in the actual proof).

If on the other hand, one of the $Y_j$'s has low min-entropy, then we will restrict our attention to a subset of useful messages on which $Y_j$ is fixed. Loosely speaking, this reduces $b$ by one, while not reducing the number of useful messages by too much (because the low min-entropy condition says that the amount of information that $Y_j$ carries on $m$ is small).
In this trench warfare, in every iteration, we lose a fraction of useful messages, for the sake of decreasing $b$ by one.
Thus, eventually, we either reach the situation that all the $Y_j$'s have large min-entropy, in which case we are done, or
we reach the situation where $B(m)$ is fixed for all messages which we can also handle by the above.

We can withstand the losses and eventually win if $\eps$ is sufficiently larger than $1/\sqrt{k}$.

\section{Preliminaries}
\label{sec:prelims}

\paragraph{Relative Hamming weight and distance: }
For a string $x \in \B^n$, we use $\weight(x)$ to denote the \emph{relative Hamming weight} of $x$, namely $\weight(x)=|\set{i: x_i=1}|/n$.

For two strings $x,y \in \B^n$, we use $\mathsf{dist}(x,y)$ to denote the \emph{relative Hamming distance} between $x$ and $y$, namely $\mathsf{dist}(x,y)=|\set{i \in [n]:x_i \ne y_i}|/n$.

\subsection{Random variables}

\paragraph{Notation for random variables: }
We use $U_n$ to denote the uniform distribution on $\B^n$. 
Given a distribution $D$, we use $x \from D$ to denote the experiment in which $x$ is chosen according to $D$. For a set $S$ we also use $x \from S$ to denote the experiment in which $x$ is chosen uniformly from $S$. When we write $x_1 \from D_1,x_2 \from D_2$ we mean that the two experiments $x_1 \from D_1$ and $x_2 \from D_2$ are independent. If $X$ is a random variable, and $D$ is a distribution, then expressions of the form $\Pr_{y \from D}[\cdot]$ where the event involves both $X$ and $y$, are in a probability space where the experiments producing $X$ and $y \from D$ are independent.

\paragraph{Min-entropy: } For a discrete random variable $X$ over $\B^n$, we define the $\emph{min-entropy}$ $\Hi(X)$ of $X$ by:
\[\Hi(X)=\min_{x \in \B^n} \frac{1}{\log \Pr[X=x]}. \]

We will also use the following lemma.

\begin{restatable}{lemma}{highentguess}\label{lem:high-ent-guess}
Suppose that $M$ is a distribution over $\{0,1\}^k$ that is uniform over a subset $S$ of size $2^r$ for $r \geq k-k^{0.99}$. If $k$ is sufficiently large, then for every function $D: [k] \to \{0,1\}$, we have that:
$$\Pr_{m \leftarrow M, i \from [k]}[D(i)=m_i] \le 0.5001.$$
\end{restatable}

The proof of Lemma \ref{lem:high-ent-guess} appears in Appendix \ref{sec:app:high-ent-guess}.

\subsection{The number of queries needed to distinguish $\BSC^n_{\half-\eps}$ from $\BSC^n_{\half}$}
The following lemma by Shaltiel and Viola \cite{SV08} is a formal restatement of Informal Theorem \ref{ithm:q queries}.

\begin{lemma}[\cite{ViolaThesis,SV08}]
\label{lem:q queries}
For every $\eps,\delta>0$, such that $\delta < 0.4$, if
$T:\B^q \ar \B$ satisfies:
\begin{itemize}
\item $\Pr[T(\BSC^q_{\half-\eps})=1] \ge 1-\delta$.
\item $\Pr[T(\BSC^q_{\half})=1] \le 0.51$.
\end{itemize}
Then, $q = \Omega\left(\frac{\log \frac{1}{\delta}}{\eps^2}\right)$.
\end{lemma}

\subsection{Constant depth circuits, approximate majority, and the coin problem}

As is standard in complexity theory, when discussing circuits, we consider circuits over the standard set of gates $\set{\textsc{AND,OR,NOT}}$. We use the convention that the size of a circuit is the number of gates and wires. With this convention, a circuit $C$ that on input $x \in \B^n$, outputs $x_1$, has size $O(1)$.
If we mention the depth of the circuit, then we mean that $\textsc{AND,OR}$ gates have unbounded fan-in, and otherwise these gates have fan-in 2.

\paragraph{Constant depth circuits for approximate majority:}
We use the following classical result by Ajtai showing that constant depth circuits can compute approximate majority:

\begin{theorem}[\cite{Ajtai}]
\label{thm:Ajtai}
There exists a constant $c$ such that for every two constants $0 \le p <P <1$, and every sufficiently large $n$, there exists a circuit $C$ of size $n^c$ and depth $c$ such that for every $x \in \B^n$:
\begin{itemize}
\item If $\weight(x) \ge P$ then $C(x)=1$.
\item if $\weight(x) \le p$ then $C(x)=0$.
\end{itemize}
\end{theorem}

\paragraph{Lower bounds for the coin problem: }
We use lower bounds on the size of constant depth circuits for the ``coin problem''. A sequence of works by \cite{SV08,Aaronson10,CGR14,LSSTV19} gives such lower bounds, and the statement below is due to Aaronson \cite{Aaronson10} and Cohen, Ganor and Raz \cite{CGR14} (and was improved by Limaye et al. \cite{LSSTV19} to also hold for circuits that are allowed to use $\textsc{PARITY}$ gates of unbounded fan-in).

\begin{theorem}[\cite{Aaronson10,CGR14,LSSTV19}]\label{thm:coin-problem}
Suppose $C: \B^{n} \ar \B$ is a circuit of depth $d$ satisfying:
\begin{itemize}
\item $\Pr_{z \from \BSC^{n}_{\half-\eps}}[C(z)=1] \ge 0.9$,
\item $\Pr_{z \from \BSC^{n}_{\half + \eps}}[C(z)=1] \le 0.1$.
\end{itemize}
Then, $C$ must have size at least $\exp(\Omega\left(d \cdot (1/\eps)^{\frac{1}{d-1}}\right))$.
\end{theorem}

For our purposes, we prefer to replace the distributions $\BSC^{n}_{\half-\eps}$ and $\BSC^{n}_{\half+\eps}$, by $\BSC^{n}_{\half-\eps}$ and $\BSC^{n}_{\half}$ (as is the case in Lemma \ref{lem:q queries}).
The next corollary shows that the results of Theorem \ref{thm:coin-problem} imply a similar bound when comparing $\BSC^{n}_{\half-\eps}$ to $\BSC^{n}_{\half}$.\footnote{We remark that bounds for the latter choice of distributions were proven by Shaltiel and Viola \cite{SV08}, but we prefer to rely on the subsequent bounds of \cite{Aaronson10,CGR14,LSSTV19}, which are tighter and
lead to a larger constant in the exponent of $\frac{1}{\eps}$ in Theorem \ref{thm:main:code:small eps}.}

\begin{restatable}{corollary}{coin}\label{cor:coin-problem}
Suppose $C: \B^{n} \ar \B$ is a circuit of depth $d$ satisfying:
\begin{itemize}
\item $\Pr_{z \from \BSC^{n}_{\half-\eps}}[C(z)=1] \ge 0.99$,
\item $\Pr_{z \from \BSC^{n}_{\half}}[C(z)=1] \le 0.01$.
\end{itemize}
Then, $C$ must have size at least $\exp(\Omega\left(d \cdot (1/\eps)^{\frac{1}{d-1}}\right))$.
\end{restatable}

For completeness, we show that Corollary \ref{cor:coin-problem} follows from Theorem \ref{thm:coin-problem} in Appendix \ref{sec:app:coin problem}.

\section{Query complexity lower bounds for local list-decoding}
\label{sec:LLDC}

In this section we prove Theorem \ref{thm:main:code:large eps} and Theorem \ref{thm:main:code:small eps} and provide lower bounds on the query complexity of local list-decoders.
In Section \ref{sec:ARLLD}, we introduce a relaxed concept that we call ``approximate local list-decoders on noisy random codewords'' (ARLLD) in which the local list-decoder is only required to recover a random message that was corrupted by a binary symmetric channel (meaning that the original message appears in the list of messages that are locally computed by the local list-decoder).
We then use a careful averaging argument to show that any local list-decoder (LLD) can be converted into an ARLLD with roughly the same parameters, and furthermore, the obtained ARLLD is \emph{deterministic}.
This means that when proving Theorem \ref{thm:main:code:large eps} and Theorem \ref{thm:main:code:small eps} it is sufficient to consider ARLLDs, and these proofs appear in Sections \ref{sec:prf:large eps} and \ref{sec:prf:small eps} respectively.

\subsection{Definition of approximate local list-decoders on noisy random codewords}
\label{sec:ARLLD}

Our goal is to prove lower bounds on the number of queries $q$ of $(\half-\eps,L,q,\delta)$-local list-decoders. We will show that it is sufficient to consider local list-decoders that need to perform an easier task. More specifically, we relax the task of a local list-decoder in the following ways:
\begin{itemize}
\item The local list-decoder does not need to succeed on every $w \in \B^n$, but only with not too small probability
 over a ``random noisy codeword'' (as defined in Definition \ref{dfn:RNSY}) which is sampled by encoding a uniformly chosen message $m$, and hitting $\Enc(m)$ with the noise generated by a binary symmetric channel, to obtain a word $w$.
 It is required that with not too small probability over the choice of the message $m$ and the random noise, there exists $j \in [L]$ such that the local decoder with oracle access to $w$, and input $j$, recovers $m$.
\item The local decoder is \emph{approximate} and is not required to recover $m_i$ correctly on every $i \in [k]$. Instead, it is allowed to err on a $\delta$ fraction of $i$'s.
\end{itemize}
This makes the task of the decoder easier.
It turns out that with this relaxation, random coins are not very helpful to the local list-decoder, and so, it is sufficient to consider deterministic local list-decoders (which do not have access to random coins\footnote{This is because any local list-decoder with failure probability $\delta$ can be converted into a deterministic approximate local list-decoder that errs on at most a $\delta$-fraction of the entries by fixing a random string for which the local list-decoder correctly recovers at least a $(1-\delta)$-fraction of the entries (such a string exists by averaging).}). This is captured in the following definition.

%

\begin{definition}[Approximate local list-decoder on noisy random codewords]
\label{dfn:ARLLD}
Let $\Enc:\B^k \ar \B^n$ be a function, and $\eps< \frac{1}{4}$. A \remph{$(\half-\eps,L,q,\delta)$-approximate $\RNSY$ local list-decoder (ARLLD)} for $\Enc$ is a \emph{deterministic}  oracle procedure $\Dec^{(\cdot)}$ that receives oracle access to a word $w \in \B^n$, and makes at most $q$ calls to the oracle. The procedure $\Dec$ also receives inputs:
\begin{itemize}
\item $i \in [k]$ : The index of the symbol that it needs to decode.
\item $j \in [L]$ : An index to the list.
\end{itemize}
It is required that, with probability at least $1/3$ over choosing $(m,z,w) \from \RNSY^{\Enc}_{\half-2\eps}$, there exists $ j \in [L]$ such that
\[ \Pr_{i \from [k]}[\Dec^w(i,j)=m_i] \ge 1-\delta. \]
\end{definition}

We stress again that $\Dec$ is deterministic and the probability in Definition \ref{dfn:ARLLD} is taken over the choice of a uniform random coordinate $i \in [k]$.

The following proposition shows that in order to prove lower bounds on local list-decoders (LLDs, Definition \ref{dfn:LLD}), it is sufficient to prove lower bounds on approximate $\RNSY$ local list-decoders (ARLLDs, Definition \ref{dfn:ARLLD}).

\begin{proposition}[LLD implies ARLLD]
\label{prop:LLD to ARLLD}
There exists a universal constant $a>1$ such that
for every $a \cdot \sqrt{\frac{1}{n}} \le \eps<\frac{1}{4}$,
if there exists an $(\half-\eps,L,q,\delta)$-local list-decoder for a function $\Enc:\B^k \ar \B^n$ then there also exists an $(\half-\eps,L,q,10 \cdot \delta)$-approximate $\RNSY$ local list-decoder for $\Enc$.
\end{proposition}

\begin{proof}
Within this proof, in order to avoid clutter, we use $\RNSY$ to denote $\RNSY^{\Enc}_{\half-2\eps}$.
Let $\Dec$ denote an LLD for $\Enc$.
For $(m,z,w) \from \RNSY$, by a Chernoff bound, for $\gamma=2^{-\Omega(\eps^2 \cdot n)}$, with probability $1-\gamma$, we have that $\mathsf{dist}(\Enc(m),w) \le \half-\eps$, meaning that $m \in \List^{\Enc}_{\half-\eps}(w)$. By the definition of LLD, this gives that whenever this occurs, with probability at least $2/3$ over the choice of $r^{\shared}$, there exists $j \in [L]$ such that the procedure $P_{w,j,r^{\shared}}(i,r)=\Dec^w(i,j,r^{\shared},r)$ locally computes $m$ with error $\delta$, where $r^{\shared}, r$ are the randomness strings  used by $\Dec$.

Let $E_1$ be the experiment in which $(m,z,w) \from \RNSY$ and $r^{\shared}$ is an independent uniform string. It follows that:

\[\Pr_{E_1}[\exists j \in [L]: \mbox{$P_{w,j,r^{shared}}$ locally computes $m$ with error $\delta$}] \ge \frac{2}{3} - \gamma. \]

By averaging, there exists a fixed string $\hat{r}^{\shared}$ such that:
\[\Pr_{\RNSY}[\exists j \in [L]: \mbox{$P_{w,j,\hat{r}^{shared}}$ locally computes $m$ with error $\delta$}] \ge \frac{2}{3} - \gamma. \]

Let $S$ denote the set of triplets $(m,z,w)$ in the support of $\RNSY$ for which the event above occurs. For every such triplet, we have that there exists a $j \in [L]$ for which $P_{w,j,\hat{r}^{shared}}$ locally computes $m$ with error $\delta$. Let $f$ be a mapping that given a triplet $(m,z,w) \in S$, produces such a $j \in [L]$. This means that:

\[\Pr_{\RNSY}[\mbox{$P_{w,f(m,z,w),\hat{r}^{shared}}$ locally computes $m$ with error $\delta$}] \ge \frac{2}{3} - \gamma. \]

Let $\RNSY'$ be the experiment in which $(m,z,w) \from \RNSY|(m,z,w) \in S$. Namely, we choose $(m,z,w)$ from the experiment $\RNSY$, conditioned on the event that $(m,z,w) \in S$.

Let $E_2$ be the experiment in which we choose independently a random string $r$, $i \from [k]$ and $(m,z,w) \from \RNSY'$.
We obtain that:

\[\Pr_{E_2}[\Dec^w(i,f(m,z,w),\hat{r}^{shared},r)=m_i] \ge 1-\delta, \]
since $P_{w,f(m,z,w),\hat{r}^{shared}}$ computes correctly each coordinate $m_i$ with probability at least $1-\delta$ over the choice of $r$.

By averaging, there exists a fixed string $\hat{r}$ such that:


\[\Pr_{(m,z,w) \from \RNSY', i \from [k]}[\Dec^w(i,f(m,z,w),\hat{r}^{shared},\hat{r})=m_i] \ge 1-\delta. \]

By Markov's inequality:

\[\Pr_{(m,z,w) \from \RNSY'}\left[\Pr_{i \from [k]}[\Dec^w(i,f(m,z,w),\hat{r}^{shared},\hat{r}) \ne m_i] \ge 10\delta\right] \le \frac{1}{10}. \]

Let $\overline{\Dec}^w(i,j)=\Dec^w(i,j,\hat{r}^{shared},\hat{r})$. We obtain that:

\[\Pr_{(m,z,w) \from \RNSY'}\left[\Pr_{i \from [k]}[\overline{\Dec}^w(i,f(m,z,w)) = m_i] > 1-10\delta\right] >\frac{9}{10}. \]

Which gives that:
\[\Pr_{(m,z,w) \from \RNSY}\left[\Pr_{i \from [k]}[\overline{\Dec}^w(i,f(m,z,w)) = m_i] > 1-10\delta\right] > \left(\frac{2}{3} - \gamma\right) \cdot \frac{9}{10} \ge \frac{1}{3}, \]
where the second inequality follows because by our requirements on $\eps$, we can choose $a$ so that $2/3-\gamma > \frac{1}{2}$. Thus, the oracle procedure $\overline{\Dec}^{(\cdot)}$ is a $(\half-\eps,L,q,10 \cdot \delta)$-ARLLD as required.
\end{proof}

By Proposition \ref{prop:LLD to ARLLD} in order to prove our main theorems on local list-decoders, it is sufficient to prove them for approximate $\RNSY$ local list-decoders.

\subsection{Proof of Theorem \ref{thm:main:code:large eps}}
\label{sec:prf:large eps}

In this section we prove Theorem \ref{thm:main:code:large eps}, restated below.

\largeeps*

We use the following definition.

\begin{definition}
Given a string $w \in \B^n$, a subset of coordinates $B=\set{h_1<\ldots<h_b} \subseteq [n]$ of size $b$, and a string $v \in \{0,1\}^B$, we let $\Fix_{B \ar v}(w) \in \B^n$  denote the string that is obtained from $w$ by fixing the bits in $B$ to the corresponding values in $v$. That is,
\[ (\Fix_{B \ar v}(w))_{\ell} = \left\{\begin{array}{lr}
        v({h_i}) , & \exists i \mbox{ s.t.\ } \ell = h_i \\
        w_{\ell}, & \ell \not \in B
        \end{array}\right. \]
\end{definition}

The lower bound will follow from the following lemma.

\begin{lemma}
\label{lem:exists m}
There exists a universal constant $\nu >0$ such that the following holds for any $L \leq 2^{k^{0.9}}$, $\epsilon \in (k^{-\nu}, \frac 1 4), \delta \in (  k^{-\nu}, \frac 13)$, and $ q \leq \frac {\log(1/\delta)} {\epsilon^2}$.
Let $\Dec$ be a $(\half-\eps,L,q,\delta)$-ARLLD for $\Enc:\B^k \ar \B^n$. Then
there exist $m' \in \B^k$, $i' \in [k]$, $j' \in [L]$, a subset $B \subseteq [n]$, and a string $v \in \B^{B}$ such that:
\begin{enumerate}
\item $\Pr_{z \from \BSC^n_{\half-2\eps}}\left[\Dec^{\Fix_{B \ar v}(\Enc(m') \oplus z)}(i',j')=m'_{i'}\right] \ge 1- 200 \delta$.
\item $\Pr_{z \from \BSC^n_{\half}}\left[\Dec^{\Fix_{B \ar v}(\Enc(m') \oplus z)}(i',j')=m'_{i'}\right] \le 0.51$.
\end{enumerate}
\end{lemma}
\begin{proof}[Proof of Theorem~\ref{thm:main:code:large eps}]
Consider a $(\half-\eps,L,q,\delta)$-LLD for $\Enc: \B^k \ar \B^n$.
By assumption that $\delta < \frac{1} {3}$, we can further assume that $\delta < 0.0002$, since, if otherwise, we can get to the desired error probability by amplification, at the loss of only a constant factor in the query complexity. We may further assume that $q \leq \frac {\log(1/\delta)} {\epsilon^2}$, otherwise we are done.

Applying Proposition~\ref{prop:LLD to ARLLD}, we get that there exists a $(\half-\eps,L,q,10\delta)$-ARLLD for $\Enc$.
Applying Lemma~\ref{lem:exists m} to this decoder, we can see that $\Dec$, when given oracle access to $\Fix_{B \ar v}(\Enc(m') \oplus z)$ and inputs $i',j'$ makes $q$ queries and outputs $m_{i'}$, (1) with probability at least $1 - 2000 \delta > 0.6$ if $z \from \BSC^n_{\half-\eps}$ and (2) with probability at most $0.51$ if $z \from \BSC^n_{\half}$.
Finally, viewing $\Dec^{\Fix_{B \ar v}(\Enc(m') \oplus z)}(i',j')$ as a function on at most $q$ bits of $z$ corresponding to the queries of $\Dec$ which are not in $B$, we can apply Lemma~\ref{lem:q queries}, completing the proof of Theorem~\ref{thm:main:code:large eps}.
\end{proof}

We will prove Lemma \ref{lem:exists m} using the probabilistic method. The main technical part of the proof is the following lemma.
Loosely speaking, the lemma says that if $\Dec$ is an ARLLD, then it can be used to distinguish between (versions of) $\BSC^n_{p=\half-2\eps}$ and $\BSC^n_{p=\half}$, in the following sense:
Recall that the experiment $\RNSY_p$ (which is the experiment on which the success of an ARLLD is measured) consists of choosing a uniform message $m \from \B^k$, a noise vector $z \from \BSC^n_p$, and setting $w=\Enc(m) \oplus z$.
We will now the consider the following modification of $\RNSY_p$:
\begin{itemize}
\item Rather than choosing $m \from \B^k$ , we choose $m$ from some specific distribution $\MU$. (The lemma claims that such a distribution $\MU$ exists).
\item The string $w$, will be modified in a set of indices $B \subseteq [n]$ to some value $v\in \B^B$. More precisely, we set $w=\Fix_{B(m) \ar v(m)}(\Enc(m) \oplus z)$ where the existence of suitable functions $B,v$ is claimed in the lemma.
\end{itemize}
We will refer to this modified experiment as $\WBV$ (to denote that $w$ is modified using $B$ and $v$).
Loosely speaking the lemma shows that $\Dec^w$ (where $w$ is chosen as explained above) distinguishes the case that $p=\half-2\eps$ from $p=\half$.
The precise statement of the lemma appears next.

\begin{lemma}
\label{lem:wbv}
There exists a universal constant $\nu >0$ such that the following holds for any $L \leq 2^{k^{0.9}}$, $\epsilon \in (k^{-\nu}, \frac 1 4), \delta \in (  k^{-\nu}, \frac 13)$, and $ q \leq \frac {\log(1/\delta)} {\epsilon^2}$.
Let  $\Dec$ be a $(\half-\eps,L,q,\delta)$-ARLLD for $\Enc:\B^k \ar \B^n$.
Then there exist:
\begin{itemize}
\item $j' \in [L]$,
\item Functions $B,v$ that given $m \in \B^k$ produce a set $B(m) \subseteq [n]$
    and $v(m) \in \B^{B(m)}$, respectively,
\item A distribution $\MU$ over $\B^k$,
\end{itemize}
such that if we use $\WBV_{p}$ to denote the experiment in which:
\begin{itemize}
\item A message $m \in \B^k$ is chosen by $m \from \MU$.
\item A noise string $z$ is chosen by $z \from \BSC^n_p$.
\item A word $w$ is obtained by $\Fix_{B(m) \ar v(m)}(\Enc(m) \oplus z)$.
\end{itemize}
We have that:
\begin{enumerate}
\item $\Pr_{(m,z,w) \from \WBV_{\half-2\eps},i \from [k]}[\Dec^w(i,j')=m_i] \ge 1-2\delta$.
\item $\Pr_{(m,z,w) \from \WBV_{\half},i \from [k]}[\Dec^w(i,j')=m_i] \le 0.501$.
\end{enumerate}
\end{lemma}
Lemma \ref{lem:exists m} follows from Lemma \ref{lem:wbv} by a Markov argument as follows.
\begin{proof}[Proof of Lemma~\ref{lem:exists m}]
By applying Markov's inequality to the first and second conditions in Lemma~\ref{lem:wbv}, we have:
\begin{align*}
\Pr_{m \from \MU, i \from [k]} \left[\Pr_{z\from \BSC^n_{\half - 2\epsilon}}[\Dec^{\Fix_{B(m) \ar v(m)}(\Enc(m) \oplus z)}(i,j')\neq m_i] > 200\delta\right] < \frac{1} {100},
\end{align*}
and
\begin{align*}
\Pr_{m \from \MU, i \from [k]} \left[\Pr_{z\from \BSC^n_{\half}}[\Dec^{\Fix_{B(m) \ar v(m)}(\Enc(m) \oplus z)}(i,j')= m_i] > 0.51\right] < \frac{0.501}{0.51} < 0.985.
\end{align*}
Hence, by the union bound, it follows that there exist $m' \in \{0,1\}^k, i' \in [k]$ such that:
\begin{align*}
\Pr_{z\from \BSC^n_{\half - 2\epsilon}}[\Dec^{\Fix_{B(m') \ar v(m')}(\Enc(m') \oplus z)}(i',j') = m'_{i'}] &\ge 1 - 200 \delta,\\
\Pr_{z\from \BSC^n_{\half}}[\Dec^{\Fix_{B(m') \ar v(m')}(\Enc(m') \oplus z)}(i',j')= m'_{i'}] &\le 0.51.
\end{align*}
Lemma~\ref{lem:exists m} follows.
\end{proof}

\subsubsection{Proof of the first item of Lemma \ref{lem:wbv} }

We are given a $(\half-\eps,L,q,\delta)$-ARLLD $\Dec$ for $\Enc:\B^k \ar \B^n$.
To avoid clutter, we will omit $\Enc$ in $\RNSY^{\Enc}_p$ in this section. We start with a couple of useful definitions.

\begin{definition}
We say that an element $j \in [L]$ is \remph{decoding} for $m,w$ if \[ \Pr_{i \from [k]}[\Dec^w(i,j)=m_i] \ge 1-\delta. \]
\end{definition}

The definition of ARLLD says that with probability at least $1/3$ over choosing $(m,z,w) \from \RNSY_{\half-2\eps}$, there exists a $j \in [L]$ that is decoding for $m,w$. By averaging over the $L$ choices of $j$, it follows that, there exists a $j' \in [L]$ such that  with probability at least $1/(3L)$ over choosing $(m,z,w) \from \RNSY_{\half-2\eps}$, this fixed $j'$ is decoding for $m,w$. This is stated in the next claim.

\begin{claim}
\label{clm:j'}
There exists $j' \in [L]$ such that
with probability at least $1/(3L)$ over choosing $(m,z,w) \from \RNSY_{\half-2\eps}$, $j'$ is decoding for $m,w$.
\end{claim}

For the rest of this section, fix an index $j' \in [L]$ for which the above claim holds.

\begin{definition}
We say that a message $m \in \B^k$ is \emph{useful} if
\[ \Pr_{z \from \BSC^n_{\half-2\eps}}[\mbox{$j'$ is decoding for $m,\Enc(m) \oplus z$}] \ge \frac{1}{6L} .\]
\end{definition}

It follows that:

\begin{claim}\label{clm:number-useful}
There are at least $2^k/(6L)$ useful messages.
\end{claim}

\begin{proof}
Otherwise, when choosing $m \from \{0,1\}^k ,z \from \BSC^n_{\half-2\eps}$ and setting $w=\Enc(m) \oplus z$ (as is done for $(m,z,w) \from \RNSY_{\half-2\eps}$):
\begin{align*}
\Pr[\mbox{$j'$ is decoding for $m,w$}] &\le \Pr[\mbox{$m$ is useful}] + \Pr[\mbox{$j'$ is decoding for $m,w$}|\mbox{$m$ is not useful}] \\
&<  \frac{1}{6L} + \frac{1}{6L} = \frac{1}{3L},
\end{align*}
which contradicts Claim \ref{clm:j'}.
\end{proof}

\begin{definition}
For a random variable $W$ over $\B^n$, a set $B \subseteq [n]$ and $v \in \B^{B}$, such that $\Pr[W_{h}=v(h)~~\forall h \in B] > 0$, we define the probability distribution $\Cond_{B \ar v}(W)$ to be $(W|W_{h}=v(h)~~\forall h \in B)$.
\end{definition}

\begin{remark}
For a random variable $W$ over $\B^n$, it is important to distinguish $\Fix_{B \ar v}(W)$ from $\Cond_{B \ar v}(W)$. The former means that we sample $w \from W$ and replace the content of $w$ in the indices in $B$ by the corresponding values taken from $v$. The latter is only defined if $W$ is a random variable for which the event $\set{W_h=v(h)~~\forall h \in B}$ can occur, and for such a variable, $\Cond_{B \ar v}(W)$ is obtained by \emph{conditioning} the random variable $W$ on the event $\set{W_h=v(h)~~\forall h \in B}$. In particular, this conditioning may mean that when restricting $\Cond_{B \ar v}(W)$ and $W$ to indices that are not in $B$, we may get different distributions. This is in contrast to $\Fix_{B \ar v}(W)$ where by definition, restricting $\Fix_{B \ar v}(W)$ and $W$ to indices that are not in $B$, gives the same distribution.

A useful observation is that if $W$ is a sequence of $n$ independent bit variables, then for every $B,v$, $\Cond_{B \ar v}(W)=\Fix_{B \ar v}(W)$.
\end{remark}

We shall use the following lemma from \cite{S20}, which improves a similar lemma (with more conditions) that was proven in \cite{GSV18}.

\begin{lemma}[\cite{S20}]
\label{lem:fixed}
Let $W$ be a probability distribution over $\B^n$, let $A \subseteq \B^n$ be an event such that $\Pr[W \in A] \ge 2^{-a}$, and let $W'=(W|W \in A)$. For every $\eta>0$, there exist a set $B \subseteq [n]$ of size $b=O(qa/\eta)$, and $v \in \B^{B}$ such that for every oracle procedure $D^{\cdot}$ that makes $q$ queries:
\[ |\Pr[D^{\Cond_{B \ar v}(W)}=1] - \Pr[D^{\Cond_{B \ar v}(W')}=1]| \le \eta. \]
\end{lemma}

We now explain why this lemma is useful. Note that if we start with some distribution $W$ over $\B^n$, then after conditioning on the event $\set{W \in A}$, the bits in the obtained distribution $W'=(W|W \in A)$ may become correlated. The Lemma says that there exist a set $B \subseteq [n]$ and $v \in \B^B$ such that if we \emph{further condition} both $W$ and $W'$ on the event $\set{W_{h}=v(h)~~\forall h \in B}$, to obtain the distributions $\Cond_{B \ar v}(W)$ and $\Cond_{B \ar v}(W')$, then these two distributions are ``similar'' in the sense that a procedure $D$ that makes few oracle calls, cannot significantly distinguish between them.

This is useful because if $W=\BSC^n_p$, then $W$ is a sequence of independent bits, and so, $\Cond_{B \ar v}(W)=\Fix_{B \ar v}(W)$. Namely, a distribution in which the bits in $B$ are fixed, and the bits outside of $B$ are independent and distributed like $\BSC^{n-b}_p$. Loosely speaking, this means that as long as we do not mind to condition on the event $\set{W_{h}=v(h)~~\forall h \in B}$, then in order to understand how $D$ behaves when given oracle to $W'$ it is sufficient to understand how it behaves when given oracle access to $W$.

\begin{definition}
For a message $m$ we denote by $\NSY(m)$ the distribution over $\B^n$ obtained by choosing $z \from \BSC^n_{\frac 1 2 - 2\epsilon}$, and setting $w=\Enc(m) \oplus z$. We use $\NSY'(m)$ to denote the distribution in which $w \from \NSY(m)$ conditioned on the event $\set{\mbox{$j'$ is decoding for $m,w$}}$.
\end{definition}

Using the above Lemma \ref{lem:fixed} we obtain the following.

\begin{claim}
\label{clm:fixed}
For every useful $m \in \B^k$ there exist a set $B(m) \subseteq [n]$ of size at most $b=O(q \cdot (\log L)/\delta)$, and $v(m) \in \B^{B(m)}$ such that for every $i \in [k]$:
\[ |\Pr[\Dec^{\Cond_{B(m) \ar v(m)}(\NSY(m))}(i,j')=1] - \Pr[\Dec^{\Cond_{B(m) \ar v(m)}(\NSY'(m))}(i,j')=1]| \le \delta. \]
\end{claim}

\begin{proof}
Apply Lemma \ref{lem:fixed} with $W$ being the distribution $\NSY(m)$, $A$ being the event that $j'$ is decoding for $m,w$, where $w \leftarrow \NSY(m)$, and
$D = \Dec^{( \cdot )}(i,j')$, for any $i \in [k]$. Note that indeed, under this setting we have that $$\Pr[W \in A] =
\Pr_{z \from \BSC^n_{\half-2\eps}}[\mbox{$j'$ is decoding for $m,\Enc(m) \oplus z$}] \ge \frac{1}{6L},$$ by assumption that $m$ is useful, and $W'= \NSY'(m)$.
\end{proof}

Next observe that by the definition of usefulness, we have that:

\begin{claim}
For every useful $m \in \B^k$,
\[ \Pr_{i \leftarrow [k]}[\Dec^{\Cond_{B(m) \ar v(m)}(\NSY'(m))}(i,j')=m_i] \ge 1-\delta. \]
\end{claim}

By Claim \ref{clm:fixed} for every fixed $i$, $\Dec^{(\cdot)}(i,j')$ cannot distinguish between the oracles in Claim \ref{clm:fixed} with advantage larger than $\delta$. It follows that it cannot do this when $i \from [k]$ is chosen at random, which gives:

\begin{claim}\label{clm:good decoding for useful m}
For every useful $m \in \B^k$,
\[ \Pr_{i \leftarrow [k]}[\Dec^{\Cond_{B(m) \ar v(m)}(\NSY(m))}(i,j')=m_i] \ge 1-2 \cdot \delta. \]
\end{claim}

Moreover, by definition $\NSY(m)$ is composed of $n$ independent bit random variables, and so, \[ \Cond_{B(m) \ar v(m)}(\NSY(m))=\Fix_{B(m) \ar v(m)}(\NSY(m)).\]
As Claim~\ref{clm:good decoding for useful m} is true for every useful $m$, it is also true for every probability distribution $\MU$ over useful messages $m$. This is stated below.

\begin{claim}
For any distribution $\MU$ over useful messages,
\[ \Pr_{m \from \MU,i \from [k]}[\Dec^{\Fix_{B(m) \ar v(m)}(\NSY(m))}(i,j')=m_i] \ge 1-2 \cdot \delta. \]
\end{claim}

Note that for any choice of distribution $\MU$, the experiment in which we choose $(m,z,w) \from \WBV_{\half-2\eps}$ and consider the pair $(m,w)$ is by definition identical to the experiment in which we choose $m \from \MU$ and set $w=\Fix_{B(m) \ar v(m)}(\NSY(m))$.

It follows that for our choices of $j',B(\cdot),v(\cdot)$, every distribution $\MU$ over useful messages satisfies the first item of Lemma \ref{lem:wbv}. This is summarized in the claim below.

\begin{claim}
\label{clm:wbv 1st item}
There exist:
\begin{itemize}
\item $j' \in [L]$,
\item Functions $B,v$ that given $m \in \B^k$ produce a set $B(m) \subseteq [n]$ of size at most $b=O(\frac{q \cdot \log L}{\delta})$ and $v(m) \in \B^{B(m)}$, respectively,
\end{itemize}
such that for every distribution $\MU$ over useful messages,
\[ \Pr_{(m,z,w) \from \WBV_{\half-2\eps},i \from [k]}[\Dec^w(i,j')=m_i] \ge 1-2\delta. \]
\end{claim}

\subsubsection{Proof of the second item of Lemma \ref{lem:wbv} }

Let $j',b, B(\cdot),v(\cdot)$ be as in Claim \ref{clm:wbv 1st item}.
In order to complete the proof of Lemma \ref{lem:wbv} we need to show that for these choices, there exists a distribution $\MU$ over useful messages such that:
\[ \Pr_{(m,z,w) \from \WBV_{\half},i \from [k]}[\Dec^w(i,j')=m_i] \le 0.501.\]

For $p=\half$, $\BSC^n_p$ is a uniformly chosen string of length $n$. It follows that for every $m \in \B^k$, the distributions $\Enc(m) \oplus \BSC^n_{\half}$ and $\BSC^n_{\half}$ are identical (as the uniform string $\BSC^n_{\half}$ masks out $\Enc(m)$).
This means that for every choice of distribution $\MU$, the pair $(m,w) \from \WBV_{\half}$ is distributed exactly like a pair $(m,\Fix_{B(m) \ar v(m)}(z))$ where $m \from \MU, z \from \BSC^n_{\half}$.
It follows that in order to complete the proof of Lemma \ref{lem:wbv} it is sufficient to prove the following lemma:

\begin{lemma}\label{lem:advice-hidden}
There exists a universal constant $\nu >0$ such that the following holds for any $L \leq 2^{k^{0.9}}$, $\epsilon \in (k^{-\nu}, \frac 1 4), \delta \in (  k^{-\nu}, \frac 13)$, and $ q \leq \frac {\log(1/\delta)} {\epsilon^2}$.
There exists a distribution $\MU$ over useful messages such that
for every oracle procedure $D^{(\cdot)}(i)$ that makes at most $q$ queries to its oracle it holds that:
\[ \Pr_{m \from \MU,z \from \BSC^n_{\half},i \from [k]}[D^{\Fix_{B(m) \ar v(m)}(z)}(i)=m_i] \le 0.501.\]
\end{lemma}

Lemma~\ref{lem:advice-hidden} implies Lemma~\ref{lem:wbv} by setting $D(\cdot)$ to be $\Dec(\cdot,j')$.
In order to prove Lemma \ref{lem:advice-hidden}, we will prove the following claim.

\begin{claim}
\label{clm:cases}
Let $S$ be a subset of useful messages, such that $|S| \ge 2^{k-k^{0.99}}$,
and let $\MD$ be the uniform distribution over $S$.
Let $b$ be an integer, and let $B,v$ be functions that given $m \in S$ produce a set $B(m) = \set{h_1(m)<\ldots<h_b(m)}\subseteq [n]$ of size $b$ and $v(m) \in \B^{B(m)}$, respectively.
If there exists
an oracle procedure $D^{(\cdot)}$ that makes at most $q$ queries such that:
\[ \Pr_{m \from \MD,z \from \BSC^n_{\half},i \from [k]}[D^{\Fix_{B(m) \ar v(m)}(z)}(i)=m_i] > 0.501,\]
then there exist:
\begin{itemize}
\item A subset $\bar{S} \subseteq S$ such that $\frac{|\bar{S}|}{|S|} \ge 2^{-(t+1)}$, where $t = 11 + \log b + \log q.$ 
\item An index $j \in [b]$, a codeword index $h' \in [n]$, and a value $v' \in \B$ such that for every message $m \in \bar{S}$, the $j$-th codeword index in $B(m)$ is $h_j(m)=h'$ and the value of the corresponding coordinate in $v$ is $(v(m))(h')=v'$.
\end{itemize}
\end{claim}
We first show that Claim~\ref{clm:cases} implies Lemma~\ref{lem:advice-hidden}.

\begin{proof}[Proof of Lemma~\ref{lem:advice-hidden}] Our goal is to find a distribution $\MU$ over useful messages so that
\begin{equation}\label{eq:advice-hidden}
\Pr_{m \from \MU,z \from \BSC^n_{\half},i \from [k]}[D^{\Fix_{B(m) \ar v(m)}(z)}(i)=m_i] \le 0.501
\end{equation}
for every $q$-query oracle procedure $D^{(\cdot)}(i)$.

To this end, we observe that if it is the case that $B(m) = B(m')$ and $v(m) = v(m')$ for all $m, m' \gets \MU$, then $\Fix_{B(m) \ar v(m)}(z)$ does not convey any information on $m \gets \MU$.
Thus, if $\MU$ satisfies this property, and is
 uniform over a set of size at least $2^{k-k^{0.99}}$,
then by Lemma \ref{lem:high-ent-guess}, any oracle algorithm $D$ satisfies (\ref{eq:advice-hidden}),
irrespective of the number of queries that $D$ makes.

To use this observation, we apply Claim \ref{clm:cases} repeatedly till we either reach a distribution $\MU$ that satisfies (\ref{eq:advice-hidden}), in which case we are done, or we reach a distribution $\MU$ which satisfies that $B(m)=B(m')$ and $v(m)=v(m')$ for all $m,m' \gets \MU$.

A technicality is that Claim \ref{clm:cases} assumes that for every $m$, the size of $B(m)$ is precisely $b$. In contrast, when we begin, we only know that the size of $B(m)$ is at most $b$. To solve this technicality we observe that if the initial distribution $\MU$ does not satisfy equation (\ref{eq:advice-hidden}), then for every $m$ such that $|B(m)|<b$ we can add additional $|b-B(m)|$ elements to $B(m)$, and by averaging (using the fact that $m$ is independent of $z,i$) we can choose the value of $v(m)$ at these indices, in a way that the probability in equation (\ref{eq:advice-hidden}) is not decreased. Overall, this means that at the start of the iterative process, we can assume that for every $m$, the size of $B(m)$ is precisely $b$.

We now go over the iterative argument in more detail.
For a subset of messages $S \subseteq \{0,1\}^k$, we let $J_{S,\mathrm{fixed}} \subseteq [b]$ denote the subset of all indices $j \in [b]$ for which there exist a codeword index $h' \in [n]$ and a value $v' \in \{0,1\}$ such that for every message $m \in S$, $h_j(m) = h'$ and $(v(m))(h')=v'$. Note that $|J_{S,\mathrm{fixed}}| =b$ if and only if $B(m) = B(m')$ and $v(m) = v(m')$ for all $m, m' \in S$. Let also $ B_{S}(m) =  \{h_j(m) \in B(m) \mid j \notin J_{S,\mathrm{fixed}}\}$, and $v_S(m) = v(m)|_{B_S(m)}$.

We initialize $\MU$ with the uniform distribution on the set $S$ of all useful messages.
At each step, given a distribution $\MU$ that is uniform over a set $S$ of useful messages of size at least $2^{k-k^{0.99}}$,
if there exists a $q$-query oracle procedure $D$ which does not satisfy (\ref{eq:advice-hidden}), and $|J_{S,\mathrm{fixed}}| <b$, we apply Claim~\ref{clm:cases} with the functions $B_S(m)$ and $v_S(m)$, and an oracle procedure $D_S$, defined as follows. The oracle procedure $D_S$ is identical to $D$, except that when querying an input $h' = h_j(m) \in [n]$ so that $j \in J_{S,\mathrm{fixed}}$, it assumes that the queried value is the unique $v' \in \{0,1\}$  so that $(v(m))(h')=v'$ for all $m \in S$ (note that this value is independent of $m$).

Noting that
$$
\Pr_{m \from \MU,z \from \BSC^n_{\half},i \from [k]}\left[D^{\Fix_{B(m) \ar v(m)}(z)}(i)=m_i\right]
$$
$$
= \Pr_{m \from \MU,z \from \BSC^n_{\half},i \from [k]}\left[D_S^{\Fix_{B_S(m) \ar v_S(m)}(z)}(i)=m_i\right]  > 0.501,
$$
Claim~\ref{clm:cases} implies the existence of a subset $\bar{S} \subseteq S$ such that $\frac{|\bar{S}|}{|S|} \ge 2^{-(t+1)}$, and an index $j \in [b] \setminus J_{S,\mathrm{fixed}}$, a codeword index $h' \in [n]$, and a value $v' \in \B$ such that for every message $m \in \bar{S}$, the $j$-th codeword index in $B_S(m)$ is $h_j(m)=h'$ and the value of the corresponding coordinate in $v_S$ is $(v_S(m))(h')=v'$.
We thus set $\MU$ to be the uniform distribution over the messages in $\bar{S}$. As this fixes a new position in $B(\cdot)$ and the corresponding value in $v(\cdot)$ for all messages sampled from $\bar S$, we have that $|J_{S,\mathrm{fixed}}|<|J_{\bar S,\mathrm{fixed}}|$.

Repeatedly applying the above, we eventually either reach a uniform distribution $\MU$ over a set $S$ which satisfies (\ref{eq:advice-hidden})
 for any $q$-query oracle procedure $D$, or we reach a distribution for which $|J_{S,\mathrm{fixed}}|=b$, and so
 $B(m) = B(m') $ and $v(m) = v(m')$ for all $m, m' \gets \MU$. In the former case we are clearly done, while in the latter case, by Lemma \ref{lem:high-ent-guess}, it suffices to show that when the process terminates, $\MU$ is distributed uniformly over a set of size at least $2^{k-k^{0.99}}$ (as in this case
(\ref{eq:advice-hidden}) holds for any oracle procedure $D$, irrespective of the number of queries it makes).

To see that the above condition holds, note that the total number of iterations is at most $b$, since in each iteration at least one of the $b$ indices in $B(\cdot)$ is fixed. Also recall that by Claim~\ref{clm:number-useful}, there are at least $2^k/(6L)$ useful messages. Consequently, when the process terminates, the number of messages in the support of $\MU$ is at least
\begin{eqnarray*}
\frac{2^k} {6L} \cdot 2^{-(t+1)b} &&  =   \frac{2^k} {6L} \cdot \left(2^{12} \cdot b\cdot q\right)^{-b} \\
&& \geq \frac{2^k} {L} \cdot \left(\frac{ \delta} {q^2 \log L} \right)^{O(q (\log L)/\delta)} \\
&& \geq \frac{2^k} {L} \cdot \left(\frac{  \epsilon \delta} { \log L} \right)^{O(\log(1/\delta) (\log L)/ (\delta \epsilon^2))} \\
&& = 2^k \cdot \exp \left(- (\log L \log \log L ) \cdot \poly( 1/\delta, 1/\epsilon)\right),
\end{eqnarray*}
where the first equality follows recalling that $t = 11+\log b + \log q$ by Claim \ref{clm:cases}, the second inequality follows
recalling that $b = O(q \cdot (\log L)/\delta)$ by Claim \ref{clm:wbv 1st item}, and the third inequality follows by assumption that
$q \leq \log(1/\delta)/\epsilon^2$. Finally, note that by choosing a sufficiently small constant $\nu >0$, and recalling our assumption that $L \leq 2^{k^{0.9}}$ and
$\epsilon, \delta \geq k^{-\nu}$, we can guarantee that the above expression is at least $2^{k- k^{0.99}}$. This concludes the proof of the lemma.
\end{proof}

Claim \ref{clm:cases} will follow from the next two claims:

\begin{claim}
\label{clm:high entropy}
Suppose that $\MD$ is a uniform distribution over a set $S$ of size at least $2^{k-k^{0.99}}$, and that
 for every $j \in [b]$, $\Hi(h_j(\MD)) \ge t$ for $t = 11+\log b + \log q$. Then
\[ \Pr_{m \from \MD,z \from \BSC^n_{\half},i \from [k]}\left[D^{\Fix_{B(m) \ar v(m)}(z)}(i)=m_i\right] \le 0.501.\]
\end{claim}

\begin{proof}
Let $E_B$ denote the event that $D^{\Fix_{B(m) \ar v(m)}(z)}(i)$ makes a query into $B(m)$, and let $\tilde E_B$ denote the event that $D^z(i)$ makes a query into $B(m)$. Then we have that when choosing $m \from \MD$,$z \from \BSC^n_{\half}$, and $i \from [k]$,
$$
 \Pr\left[D^{\Fix_{B(m) \ar v(m)}(z)}(i)=m_i\right]
  \leq \Pr\left[ \left(D^{\Fix_{B(m) \ar v(m)}(z)}(i)=m_i\right) \cap \neg E_B\right] + \Pr[E_B].
  $$

  To bound the right-hand term, we first claim that $\Pr[E_B]= \Pr[\tilde E_B]$. To see this, note that for any fixed $m, i$, the set of strings $z$ on which $D^{\Fix_{B(m) \ar v(m)}(z)}(i)$ makes a query into $B(m)$ is identical to the set of  strings $z$ on which  $D^z(i)$ makes a query into $B(m)$, since the locations of the queries made before the first query to $B(m)$ are the same for the oracles $\Fix_{B(m) \ar v(m)}(z)$ and $z$. Thus, to bound the right-hand term, it suffices to bound the probability of the event $\tilde E_B$.

 To bound the probability that $\tilde E_B$ occurs, we this time fix the string $z$ and the index $i$, and note that this determines the query pattern of $D^z(i)$. Next we recall our assumption that for $m \gets \MD$,  $\Hi(h_j(m)) \ge t$ for all $j \in [b]$. Thus for any $j \in [b]$, the probability, over $m \gets \MD$ (noting that this choice of $m$ is independent of the fixing of $z, i$), that a specific query of $D^z(i)$
 is to $h_j(m)$ is at most $2^{-t}$.
Hence, by a union bound, the probability that the event $\tilde E_B$ occurs is at most $q \cdot b \cdot 2^{-t}$.
Finally, by our setting of $t = 11 + \log b + \log q$, we have that this probability is at most $2^{-11}$. We conclude that the right-hand term satisfies $\Pr[E_B] \leq 2^{-11}.$



To bound the left-hand term, we once more claim that $$\Pr\left[ \left(D^{\Fix_{B(m) \ar v(m)}(z)}(i)=m_i\right) \cap \neg E_B\right]= \Pr\left[ \left(D^{z}(i)=m_i\right) \cap \neg \tilde E_B\right].$$
Once more, this follows since when fixing $m,i$, the set of strings $z$ on which $D^{\Fix_{B(m) \ar v(m)}(z)}(i)$ does not make a query into $B(m)$ is identical to the set of  strings $z$ on which  $D^z(i)$ does not make a query into $B(m)$, and fixing each such string $z$ induces
the same query pattern for $D^{z}(i)$ and $D^{\Fix_{B(m) \ar v(m)}(z)}(i)$. Thus we conclude that for fixed $m,i$, the set of strings $z$ which lead to the event $\left(D^{\Fix_{B(m) \ar v(m)}(z)}(i)=m_i\right) \cap \neg E_B$ is identical to the set of strings $z$ that
lead to the event $\left(D^{z}(i)=m_i\right) \cap \neg \tilde E_B$, and so the probabilities are the same. It thus suffices to bound the probability of the event $\left(D^{z}(i)=m_i\right) \cap \neg \tilde E_B.$

To bound the probability that $\left(D^{z}(i)=m_i\right) \cap \neg \tilde E_B$ occurs, we note that it is at most the probability that
$D^z(i)=m_i$ occurs. Recalling our assumption that $\MD$ is uniform over a set of size at least $2^{k- k^{0.99}}$, by
Lemma \ref{lem:high-ent-guess}, this latter probability is at most $0.5001$. So the left-hand term satisfies that
$$\Pr\left[ \left(D^{\Fix_{B(m) \ar v(m)}(z)}(i)=m_i\right) \cap \neg E_B\right] \leq 0.5001.$$

Summing up the two probabilities, we get that
$$ \Pr_{m \from \MD,z \from \BSC^n_{\half},i \from [k]}\left[D^{\Fix_{B(m) \ar v(m)}(z)}(i)=m_i\right] \leq 0.5001 + 2^{-11} \leq 0.501,$$
which concludes the proof of the claim.
\end{proof}

\begin{claim}
\label{clm:low entropy}
If there exists $j \in [b]$, such that $\Hi(h_j(\MD)) < t$ then there exist:
\begin{itemize}
\item A subset $\bar{S} \subseteq S$ such that $\frac{|\bar{S}|}{|S|} \ge 2^{-(t+1)}$.
\item A codeword index $h' \in [n]$ and a value $v' \in \B$ such that for every message $m \in \bar{S}$, the $j$-th codeword index in $B(m)$ is $h_j(m)=h'$ and the value of the corresponding coordinate in $v$ is $(v(m))(h')=v'$.
\end{itemize}
\end{claim}
\begin{proof}
Since $\Hi(h_j(\MD)) < t$, there must exist $h' \in [n]$ such that $\Pr[h_j(\MD) = h'] \ge 2^{-t}$.
Let $v' \in \{0,1\}$ be such that $\Pr[(v(\MD))(h') = v'| h_j(\MD) = h'] \ge 1/2$.
In other words, $v'$ is the more probable value taken by $v(m)$ at the index $h'$, conditioned on $m \gets \MD$ satisfying that the $j$-th index in $B(m)$ is $h_j(m)=h'$.
So we get that with probability at least $2^{-(t+1)}$, both events $h_j(\MD) = h'$ and $(v(\MD))(h')=v'$ hold.
This event can be thought of in turn as a subset $\bar{S}$ of messages of density at least $2^{-(t+1)}$ inside $S$.
\end{proof}

\subsection{Proof of Theorem \ref{thm:main:code:small eps}}
\label{sec:prf:small eps}

In this section we prove Theorem \ref{thm:main:code:small eps}, restated below.

\smalleps*

The Theorem will follow from the next lemma (which argues that an ARLLD with certain parameters, yields a depth 3 circuit that distinguishes between $\BSC_{\half-2\eps}$ and $\BSC_{\half}$).

Our overall plan for proving Theorem \ref{thm:main:code:small eps} is to argue that a local list-decoder with the parameters specified in Theorem \ref{thm:main:code:small eps} yields an ARLLD with the parameters specified in Lemma \ref{lem:exists C}, which in turn yields the circuit specified in Lemma \ref{lem:exists C}, and this will be used to prove a lower bound on $q$.

\begin{lemma}
\label{lem:exists C}
There exist universal constants $\beta>0$ and $c>1$ such that the following holds for any $L \le \beta \cdot 2^k$ and $\eps< \frac 1 4$.
Let $\Dec$ be a $(\half-\eps,L,q,\frac1 {2k})$-ARLLD for $\Enc:\B^k \ar \B^n$, and let $n'=cn$. Then there exists a circuit $C:\B^{n'} \ar \B$ of size $O(L \cdot k \cdot 2^{2q})$ and depth $3$, such that:
\begin{itemize}
\item $\Pr_{z \from \BSC^{n'}_{\half-2\eps}}[C(z)=1] \ge 0.99$.
\item $\Pr_{z \from \BSC^{n'}_{\half}}[C(z)=1] \le 0.01$.
\end{itemize}
\end{lemma}

We first prove that Lemma~\ref{lem:exists C} implies Theorem~\ref{thm:main:code:small eps}.
\begin{proof}[Proof of Theorem~\ref{thm:main:code:small eps}]
Consider a $(\half - \eps, L, q, \delta)$-LLD $\Dec$ for $\Enc: \B^k \to \B^n$, where $\delta \le 1/3$.
It is possible to amplify the error probability $\delta$ from $\frac 1 3$ to $\frac 1 {20k}$ as follows: After choosing the random string $r^{\shared}$, we choose $e=O(\log k)$ independent uniform strings $r_1,\ldots,r_e$, and apply $\Dec^{(\cdot)}(i,j,r^{\shared},r_{\ell})$ for all choices of $\ell \in [e]$. We then output the majority vote of the individual $e$ outputs. It is standard that this gives a $(\half - \eps, L, q'=O(q \log k),  \frac 1 {20k})$-LLD for $\Enc: \B^k \to \B^n$.

By our requirements on $\eps$, we can use Proposition~\ref{prop:LLD to ARLLD} to show that there exists a $(\half - \eps, L, q',\frac1 {2k})$-ARLLD for $\Enc: \B^k \to \B^n$.
By Lemma~\ref{lem:exists C}, there exists a circuit $C$ of size
\[s=O(L \cdot k \cdot 2^{2 \cdot q'})=L \cdot k \cdot 2^{c_1q \log k} \]
and depth $3$ that distinguishes $\BSC^{n'}_{\half-2\eps}$ from $\BSC^{n'}_{\half}$. By Corollary \ref{cor:coin-problem} such a circuit must have size:
\[ s \ge 2^{1/(c_2\sqrt{\eps})}. \]
This implies that
\[q \ge \frac{1} {c_1 c_2 \sqrt{\eps} \log k} - \frac {\log(Lk)} {c_1 \log k} \geq \frac{1} {c_1 c_2 \sqrt{\eps} \log k} - \log L . \]
This completes the proof of the theorem.
\end{proof}

In the remainder of this section we prove Lemma \ref{lem:exists C}.
Let $\Dec$ be a $(\half-\eps,L,q,\frac 1 {2k})$-ARLLD for $\Enc:\B^k \ar \B^n$.
We will construct the circuit $C$ in the following sequence of claims:

\begin{definition}
For every $i \in [k]$ and $j \in [L]$, let $A_{i,j}:\B^n \ar \B$ be defined by:
$A_{i,j}(w)=\Dec^w(i,j)$.
\end{definition}

\begin{claim}
For every $i \in [k]$ and $j \in [L]$, there exist CNF circuits $A^T_{i,j}:\B^n \ar \B$ and $A^F_{i,j}:\B^n \ar \B$ of size $O(q \cdot 2^q)$ such that for every $w \in \B^n$, $A^T_{i,j}(w)=A_{i,j}(w)$, and $A^F_{i,j}(w)=1-A_{i,j}(w)$.
\end{claim}

\begin{proof}
For fixed $i,j$, we can view the computations $A_{i,j}(w)=\Dec^w(i,j)$ as a depth $q$ decision tree that makes queries to $w$. Such a decision tree can be simulated by a size $O(q \cdot 2^q)$ DNF, which is a disjunction over the at most $2^q$ accepting paths of the tree, where each path is a conjunction of $q$ literals. This also gives a CNF of the same size for $1-A_{i,j}$. The same argument can be repeated for $1-A_{i,j}$ giving a CNF of the same size for $A_{i,j}$.
\end{proof}
%
%
%

\begin{definition}
For every $m \in \B^k$ we define the circuit $C_m:\B^n \ar \B$ that is hardwired with the message $m \in \B^k$, and the encoding $\Enc(m)$. Given input $z \in \B^n$, the circuit $C_m$ acts as follows:
\begin{itemize}
\item Prepare $w=\Enc(m) \oplus z$.
\item For every $i \in [k]$ and $j \in [L]$ compute $A_{i,j}(w)$, and compute $b_{i,j} \in \B$ which answers whether $A_{i,j}(w)=m_i$.
\item For every $j$, compute $b_j$ which is the conjunction of $b_{1,j},\ldots,b_{k,j}$.
\item Compute the disjunction of $b_1,\ldots,b_L$ and output it.
\end{itemize}
\end{definition}

\begin{claim}
\label{clm:Cm simulates Dec}
For every $m \in \B^k$ the circuit $C_m$ can be implemented in size $O(k \cdot L \cdot 2^{2q})$ and depth $3$. Furthermore, for every $m \in \B^k$, and $z \in \B^n$, $C_m(z)=1$ if and only if there exists $j \in [L]$ such that for every $i \in [k]$, $\Dec^{\Enc(m) \oplus z}(i,j)=m_i$.
\end{claim}

\begin{proof}
It is immediate that the circuit $C_m$ performs the task described in the claim. We now explain how to implement the circuit in small size and depth.

The string $m$ is of length $k$. We note that when using $\Enc(m)$ to prepare $w$, we only need to have $\Enc(m)$ at coordinates $\ell$ such that there exist $i,j$ such that $A_{i,j}(w)$ depends on the $\ell$'th input. As each circuit $A^T_{i,j}$ is a circuit of size $O(q \cdot 2^q)$ it depends on at most $O(q \cdot 2^q)$ input bits. Thus, $C_m$ only requires $O(k \cdot L \cdot q \cdot 2^q)$ bits of $\Enc(m)$. Overall, the size of the advice of $C_m$ is $O(k \cdot L \cdot q \cdot 2^q)$.

Computing every bit of $w$ amounts to at most one negation gate, and does not increase the depth.
For every $i \in [k]$ and $j \in [L]$, we want to compute the bit $b_{i,j}$ which is $1$ if and only if $A_{i,j}(w)=m_i$. Note that if $m_i=1$ then $b_{i,j}=A^T_{i,j}(w)$ and if $m_i=0$ then $b_{i,j}=A^F_{i,j}(w)$. As $m_i$ is a fixed constant, this gives that for every $i \in [k]$ and $j \in [L]$, $b_{i,j}$ can be computed by a CNF of size $O(q \cdot 2^q)$ that is applied on the input $z$.
For every $j \in [L]$, computing $b_j$ is done using a single AND gate, and, since the top gate of the CNF computing $b_i$ is also an AND gate, this does not increase the depth. Finally, computing the output adds a top OR gate.
overall, the depth is $3$ and the size is bounded by:
\[ O(k \cdot L \cdot q \cdot 2^q) \le O(k \cdot L \cdot 2^{2q}). \qedhere \]
\end{proof}

It is important to note that in Lemma \ref{lem:exists C}, we start from the assumption that $\Dec$ has $\delta=\frac{1}{2k}$ (which is achieved by reducing the initial $\delta$ at the cost of making more queries). This value of $\delta$ is smaller than $\frac{1}{k}$. By the definition of an ARLLD, this implies the following.

\begin{claim}\label{clm:Cm decodes}
$\Pr_{m \from \B^k,z \from \BSC^n_{\half-2\eps}}[C_m(z)=1] \ge \frac{1}{6}$.
\end{claim}

\begin{proof}
By the definition of an ARLLD we have that with probability at least $1/3$ over choosing $(m,z,w) \from \RNSY^{\Enc}_{\half-2\eps}$, there exists $ j \in [L]$ such that
\[ \Pr_{i \from [k]}[\Dec^w(i,j)=m_i] \ge 1-\frac{1}{2k}. \]
However, as we consider choosing a random $i \from [k]$, if the latter probability is greater than $1-\frac{1}{k}$ then it must be one. We can therefore conclude that with probability at least $1/3$ over choosing $(m,z,w) \from \RNSY^{\Enc}_{\half-2\eps}$, there exists $ j \in [L]$ such that for every $i \in [k]$, $\Dec^w(i,j)=m_i$.
By Claim \ref{clm:Cm simulates Dec} this gives that
$\Pr_{m \from \B^k,z \from \BSC^n_{\half-2\eps}}[C_m(z)=1] \ge \frac{1}{3}$ as required.
\end{proof}

On the other hand, we can show that:

\begin{claim}\label{clm:Cm does not decode}
$\Pr_{m \from \B^k,z \from \BSC^n_{\half}}[C_m(z)=1] \le L \cdot 2^{-k}$.
\end{claim}

\begin{proof}
The first step of $C_m(z)$ is to prepare $w=\Enc(m) \oplus z$.
However, for $p=\half$, and $z \from \BSC^n_p$, we have that $w=\Enc(m) \oplus z$ is uniformly chosen, and independent of $m$. This means that the bits $A_{i,j}(w)$ for $i \in [k]$ and $j \in [L]$ are independent of $m$. Consequently, for every $j \in [L]$, we have that:
\begin{align*}
& \Pr_{m \from \B^k,z \from \BSC^n_{\half},w=\Enc(m) \oplus z}[m=A_{1,j}(w)\circ \ldots \circ A_{k,j}(w)] \\
&\le \Pr_{m \from \B^k,w \from \B^n}[m=A_{1,j}(w)\circ \ldots \circ A_{k,j}(w)] \le 2^{-k}
\end{align*}
By a union bound over all choices of $j \in [L]$, we have that:
\[ \Pr_{m \from \B^k,z \from \BSC^n_{\half},w=\Enc(m) \oplus z}[\exists j:\mbox{s.t.\ } m=A_{1,j}(w)\circ \ldots \circ A_{k,j}(w)] \le L \cdot 2^{-k} \]
It follows that:
\[ \Pr_{m \from \B^k,z \from \BSC^n_{\half}}[C_m(z)=1] \le L \cdot 2^{-k}, \]
as required.
\end{proof}

We are finally ready to prove Lemma \ref{lem:exists C}.

\begin{proof}[Proof of Lemma~\ref{lem:exists C}]
By our choices, we have that $0 < L\cdot 2^{-k} \le \beta < 1$.
By applying Markov's inequality to Claims~\ref{clm:Cm decodes} and~\ref{clm:Cm does not decode}, we can see that:
\begin{align*}
\Pr_{m \from \B^k} \left[\Pr_{z \from \BSC^n_{\half-2\eps}}[C_m(z) \neq 1] > \frac{9}{10}\right] &< \frac{5/6}{9/10} \le \frac{95}{100}.\\
\Pr_{m \from \B^k} \left[\Pr_{z \from \BSC^n_{\half}}[C_m(z)=1] > \sqrt{\beta}\right] &< \sqrt{\beta}.
\end{align*}
Therefore, for a sufficiently small constant $\beta>0$, by a union bound, we get that there exists $m \in \B^k$ satisfying both:
\begin{align*}
\Pr_{z \from \BSC^n_{\half-2\eps}}[C_m(z) = 1] &\ge \frac{1}{10},\\
\Pr_{z \from \BSC^n_{\half}}[C_m(z)=1] &\le \sqrt{\beta},
\end{align*}
It is possible to amplify these thresholds to $0.99$ and $0.01$ as follows: Let $c$ be a constant that we choose later, and let $n'=cn$. Consider a circuit $C:\B^{n'} \ar \B$ that when receiving input $x \in \B^{n'}$, treats it as $c$ strings $x_1,\ldots,x_c \in \B^n$. The circuit $C$ will apply $C_m$ on each of the $c$ strings, and the final output is the OR of the results. As the top gate of $C_m$ is an OR gate, adding an additional OR gate, does not increase the depth of the circuit. The size of the circuit increases by a constant factor.

We can view $z \from \BSC^{n'}_{p}$ as obtained by concatenating $c$ strings $z_1,\dots, z_c$, where each of them is sampled uniformly and independently at random from $\BSC^n_p$.
Hence, the event $C(z) = 1$ is identical to $\bigvee_{\ell \in [c]} C_m(z_\ell) = 1$.
Therefore,
$$\Pr_{z \from \BSC^{n'}_{\half-2\eps}}[C(z) = 1] \ge 1 - \left(\frac{9}{10}\right)^c,$$
and
$$\Pr_{z \from \BSC^{n'}_{\half}}[C(z)=1] \le c\cdot\sqrt{\beta}.$$
By choosing $c$ to be sufficiently large, we can see that $1 - (\frac{9}{10})^c \ge 0.99$.
By choosing $\beta$ to be sufficiently small, we also have that $c \sqrt{\beta} \le 0.01$.
\end{proof}

\section{Limitations on black-box proofs for hard-core predicates}
\label{sec:Hard-core}

In this section, we present our results regarding the limitations on black-box proofs for hard-core predicate theorems. In Section \ref{sec:HC compute}, we state our results for functions that are hard to compute, give a formal restatement of Theorem \ref{ithm:compute}, and prove the theorem. In Section \ref{sec:HC invert}, we state our results for functions that are hard to invert, give a formal restatement of Theorem \ref{ithm:invert}, and prove the theorem.

\subsection{The case of functions that are hard to compute}
\label{sec:HC compute}

\subsubsection{The model for black-box proofs}
\label{sec:HC compute definition}

In this section, we state and explain our model for black-box proofs for hard core predicates, in the setting of functions that are hard to compute. The formal definition is given in Definition \ref{dfn:HC compute}. Below, we provide a detailed explanation for the considerations made while coming up with the formal definition. The reader can skip directly to the formal definition if they wish to.

\paragraph{Explanation of the model:}
Recall that (as explained in Section \ref{sec:intro:HC:compute}) the Goldreich-Levin theorem (stated precisely in Theorem \ref{thm:GL compute}) has the following form:
\begin{itemize}
\item We are given an arbitrary hard function $g:\B^{\ell} \ar \B^{\ell}$. (Intuitively, it is assumed that it is hard to compute $g$ with success probability $\rho$).
\item There is a specified construction that transforms $g$ into a predicate $g^{\pred}:\B^{\ell'} \ar \B$ for some $\ell'$ related to $\ell$. (Intuitively, we will want to argue that $g^{\pred}$ is a hard-core predicate that is hard to compute with success $\half+\eps$).

    We will model this construction as a map $\Con$, which, given $g$ produces $g^{\pred}$. We place \emph{no limitations} on the map $\Con$ (and, in particular, do not require that $g^{\pred}$ can be efficiently computed if $g$ is). This only makes our results stronger.

    In the case of Theorem \ref{thm:GL compute}, we have that: $\Con(g)=g^{\pred}$ where $\ell'=2\ell$ and we think of the $\ell'$-bit long input of $g^{\pred}$ as two strings $x,r \in \B^{\ell}$, setting:  \[g^{\pred}(x,r)=\Enc^{\Hadamard}(g(x))_r = (\sum_{i \in [\ell]} g(x)_i \cdot r_i ) \mod 2. \]
\item We model the proof showing that $g^{\pred}$ is a hard-core predicate in the following way: The proof is a pair $(\Con,\Red)$ where $\Red^{(\cdot)}$ is an oracle procedure, such that when $\Red^{(\cdot)}$ receives oracle access to an ``adversary'' $h:\B^{\ell'} \ar \B$ that breaks the security of $g^{\pred}$, we have that $\Red^{h}$ breaks the security of $g$. More precisely, we require that: for every $g:\B^{\ell} \ar \B^{\ell}$ and for every $h:\B^{\ell'} \ar \B$ such that:
    \[ \Pr_{x \from U_{\ell'}}[h(x)=g^{\pred}(x)] \ge \half+\eps,\]
    it holds that:
    \[\Pr_{x \from U_{\ell}}[\Red^{h}(x)=g(x)] \ge \rho.\]
\item In the actual definition, we will allow the reduction to have more power (which only makes our results stronger). As we are aiming to prove a result on circuits (which are allowed to use nonuniform advice) we will allow the reduction to receive an advice string $\alpha$ of length $t$, where, this advice string can depend on $g$ and $h$. This leads to the following strengthening of the requirement above. Namely, we will require that: for every $g:\B^{\ell} \ar \B^{\ell}$ and for every $h:\B^{\ell'} \ar \B$, such that:
    \[ \Pr_{x \from U_{\ell'}}[h(x)=g^{\pred}(x)] \ge \half+\eps,\]
    there exists $\alpha \in \B^t$ such that:
    \[\Pr_{x \from U_{\ell}}[\Red^{h}(x,\alpha)=g(x)] \ge \rho.\]
    We remark that in many related settings (for example, ``hardness amplification''; see \cite{SV08,GSV18}, for a discussion) known proofs by reduction \emph{critically make use} of the ability to introduce nonuniformity, and so, we feel that when ruling out black-box proofs in scenarios involving circuits, it is necessary to consider nonuniform black-box reductions.
\item We make no restrictions on the complexity of the procedure $\Red^{(\cdot)}$, except for requiring that it makes at most $q$ queries to its oracle (for some parameter $q$). Our black-box impossibility results will follow from proving lower bounds on $q$.
\end{itemize}

\paragraph{Formal definition:}
We now give a formal definition of our model for black-box proofs for hard-core predicates.

\begin{definition}[Nonuniform black-box proofs for hard-core predicates for hard-to-compute functions]
\label{dfn:HC compute}
A pair $(\Con,\Red)$ is a \remph{nonuniform black-box proof} for \remph{hard-core predicates for hard-to-compute functions} with parameters $\ell,\ell',\rho,\eps$, that uses \remph{$q$ queries}, and \remph{$t$ bits of advice} if:
\begin{itemize}
\item $\Con$ is a \emph{construction map} which given a function $g:\B^{\ell} \ar \B^{\ell}$, produces a function $\Con(g)=g^{\pred}$, where $g^{\pred}:\B^{\ell'} \ar \B$.
\item $\Red^{(\cdot)}$ is a \emph{reduction}, that is an oracle procedure that, given oracle access to a function $h:\B^{\ell'} \ar \B$, makes at most $q$ queries to its oracle.
\end{itemize}
Furthermore, for every functions $g:\B^{\ell} \ar \B^{\ell}$ and $h:\B^{\ell'} \ar \B$ such that:
    \[ \Pr_{x \from U_{\ell'}}[h(x)=g^{\pred}(x)] \ge \half+\eps, \]
    there exists $\alpha \in \B^t$, such that:
    \[ \Pr_{x \from U_{\ell}}[\Red^h(x,\alpha)=g(x)] \ge \rho. \]
\end{definition}

\paragraph{The role of the number of queries, and black-box impossibility results: }
We now explain the role of the parameter $q$ (that measures the number of queries made by $\Red$) and why lower bounds on $q$ translate into black-box impossibility results.

For this purpose, it is illustrative to examine the argument showing that nonuniform black-box proofs yield hard-core predicates: When given a pair $(\Con,\Red)$ that is a nonuniform black-box proof for hard-core predicates for hard-to-compute functions with parameters $\ell,\ell',\rho,\eps$, that uses $q$ queries, and $t$ bits of advice, we obtain that for any function $g:\B^{\ell} \ar \B^{\ell}$, if there exists a circuit $C':\B^{\ell'} \ar \B$ of size $s'$ such that:
\[ \Pr_{x \from U_{\ell'}}[C'(x)=g^{\pred}(x)] \ge \half+\eps, \]
    then there exists $\alpha \in \B^t$, such that:
    \[ \Pr_{x \from U_{\ell}}[\Red^{C'}(x,\alpha)=g(x)] \ge \rho. \]
Note that if the reduction $\Red$ can be implemented by a circuit of size $r$, then the circuit $C(x)=\Red^{C'}(x,\alpha)$ is a circuit of size:
\[ s=r+t+q \cdot s' \]
that computes $g$ with success probability $\rho$.

It follows that in a black-box proof, with $q$ queries, and $t$ bits of advice, we get a hard-core theorem that needs to assume that the original function $g$ has hardness against circuits of size $s$, for:
\[ s \ge q +t. \]

\subsubsection{Precise statements of limitations}

Our main result on black-box proofs for hard-core predicates in the setting of functions that are hard to compute is the following theorem.

\begin{theorem}
\label{thm:hard-core compute}
There exist universal constants $\beta>0$ and $c>1$ such that the following holds for any sufficiently large $\ell$ and $\ell'$, $\eps > 0$, $t \le 2^{\ell/5}$, and $\rho \ge 2^{-\ell/3}$.
Let $(\Con,\Red)$ be a nonuniform black-box proof for hard-core predicates for hard-to-compute functions with parameters $\ell,\ell',\rho,\eps$, that uses $q$ queries, and $t$ bits of advice. Then
\[ q \ge \frac{1}{\eps^{\beta}}- c (t+\ell). \]
\end{theorem}

We now explain why Theorem \ref{thm:hard-core compute} implies the informal statement made in Theorem \ref{ithm:compute}. Recall that in Section \ref{sec:HC compute definition} we explained that when using a nonuniform black-box proof to obtain a hard-core predicate, we get a hard-core predicate theorem in which $s \ge q+t$.

Theorem \ref{thm:hard-core compute} implies that it is impossible for such a proof to establish $\eps=s^{-2/\beta}$ (even if $\rho$ is very small). This follows as otherwise, using the fact that $s \ge q+t \ge t$ and $s \geq \ell$ (since the input to the circuit is $x \in \{0,1\}^\ell$), we get that:
\[ q \ge \frac{1}{\eps^{\beta}} -c(t + \ell) \ge s^2 -c(t+\ell)> s,\]
which is a contradiction to $s \ge q+t \ge q$. In particular, the parameter setting considered in Theorem \ref{ithm:compute}, in which  $\eps=\frac{1}{s^{\omega(1)}}$, is impossible to achieve.

\subsubsection{Proof of Theorem \ref{thm:hard-core compute}}

Theorem \ref{thm:hard-core compute} will follow from the next lemma, showing that a proof with small $q$ can be transformed into a small constant depth circuit for the coin problem.

\begin{lemma}\label{lem:proof-implies-coin-distinguish}
\label{lem:hard-core compute exists C}
There exists a universal constant $d>1$
such that the following holds
 for any sufficiently large $\ell$ and $\ell'$, $\eps \ge 2^{-\ell'/3}$, $t \le 2^{\ell/5}$, and $\rho \ge 2^{-\ell/3}$.
Let $(\Con,\Red)$ be a nonuniform black-box proof for hard-core predicates for hard-to-compute functions with parameters $\ell,\ell',\rho,\eps$, that uses $q$ queries, and $t$ bits of advice. Then there exists a circuit $C$ of size $\poly(2^q,2^{\ell},2^t)$ and depth $d$ such that:
\begin{itemize}
\item $\Pr_{z \from \BSC^{n}_{\half-2\eps}}[C(z)=1] \ge 0.99$.
\item $\Pr_{z \from \BSC^{n}_{\half}}[C(z)=1] \le 0.01$.
\end{itemize}
\end{lemma}
We first show that Theorem~\ref{thm:hard-core compute} follows from Lemma~\ref{lem:hard-core compute exists C}.

\begin{proof}[Proof of Theorem~\ref{thm:hard-core compute}]
First note that as $x$ is sampled from the uniform distribution on $\{0,1\}^{\ell'}$, we have that $\eps \geq 2^{-\ell'}$. Moreover, by definition it follows that if $(\Con,\Red)$ is a nonuniform black-box proof for hard-core predicates with parameter $\eps$, then it is also a nonuniform black-box proof for hard-core predicates with parameter $\eps^{1/3} \geq 2^{-\ell'/3}$.
Lemma \ref{lem:hard-core compute exists C} and Corollary \ref{cor:coin-problem} then give that:
\[ \poly(2^q,2^{\ell},2^t) \ge \exp\left(d \cdot \eps^{-\frac 1 {3(d-1)}}\right).\]
The statement of Theorem~\ref{thm:hard-core compute} follows by taking the logarithm on both sides and setting $\beta < \frac 1 {3(d-1)}$.
\end{proof}

In the remainder of this section we prove Lemma \ref{lem:hard-core compute exists C}.
Let $(\Con,\Red)$ be a nonuniform black-box proof for hard-core predicates for hard-to-compute functions with parameters $\ell,\ell',\rho,\eps$, that uses $q$ queries, and $t$ bits of advice. Throughout this section we assume that the requirements made in Lemma \ref{lem:hard-core compute exists C} are met.

We will identify functions $h:\B^{\ell} \ar \B$ with strings $h \in \B^{2^{\ell}}$. More precisely, we fix some ordering on strings $x \in \B^{\ell}$ and then, the value of string $h$ at position $x$ is the function $h$ applied on $x$.
We will use $h$ to denote these two objects (both the function and the string) and this means that a function $h$ can be given as an argument to a function that receives strings of length $2^{\ell}$.

\paragraph{High level description of the proof of Lemma \ref{lem:hard-core compute exists C}.}
The proof will use the same structure as the proof of Lemma \ref{lem:exists C}, which was the main technical lemma in the proof of Theorem \ref{thm:main:code:small eps}. Loosely speaking, the reduction $\Red$ plays the role of a local list-decoder, the function $g$ plays the role of the message, the construction map $\Con$ plays the role of the encoder, and the function $g^{\pred}$ (viewed as a $2^{\ell'}$ bit long string) plays the role of the encoding of the message.

Imitating the approach used in the proof of Lemma \ref{lem:exists C}, we will consider the function (or string) $h \in \B^{2^{\ell'}}$ defined by $h=g^{\pred} \oplus z$ where $z \from \BSC^{2^{\ell'}}_{p}$ for $p=\half$ and $p=\half-2\eps$. We will try to show that when the function $g$ is chosen uniformly, then for $p=\half-2\eps$, $\Red^h$ has to succeed, and for $p=\half$, $\Red^h$ cannot succeed. We will then leverage this difference to produce a constant depth circuit of size roughly $2^q$ that distinguishes $\BSC^{2^{\ell'}}_{\half-2\eps}$ from $\BSC^{2^{\ell'}}_{\half}$.

However, there are some complications. When the reduction $\Red^h$ succeeds for $p=\half-2\eps$, we only have that there exists an $\alpha$ with which it computes $g$ with success $\rho$ which is extremely small, and is not much larger than the success probability of $\Red^h$ for $p=\half$ (which we will show is less than $\rho/10$).

At first glance, this seems like a problem, as in general, in order to distinguish success probability $a+\rho/10$ from $a+\rho$ for an arbitrary value of $a \in [0,1]$, constant depth circuits need to have size that is $2^{(1/\rho)^{\Omega(1)}}$ which is much too large for our purposes. Fortunately, for $a=0$ (that is for the task of distinguishing success probability $\rho/10$ from $\rho$) it is possible to distinguish with circuits of size $\poly(1/\rho)$. The formal statement of this is given in Lemma \ref{lem:D rho}.

We now return to the formal proof, starting with the following lemma.

\begin{lemma}
\label{lem:D rho}
There exists a universal constant $d>1$, such
that any $n,\rho$ there exists a circuit $\D^n_{\rho}:\B^n \ar \B$ of size $\poly(n/\rho)$ and depth $d$, such that for every $x \in \B^n$:
\begin{itemize}
\item If $\weight(x) \ge \rho$ then $\D^n_{\rho}(x)=1$.
\item If $\weight(x) \le \rho/10$ then $\D^n_{\rho}(x)=0$.
\end{itemize}
\end{lemma}

We note that by the lower bound of Razborov and Smolensky \cite{Razborov,Smolensky} small constant-depth circuits cannot compute the majority function. Nevertheless, by the results of Ajtai \cite{Ajtai} (stated in Theorem \ref{thm:Ajtai}) small constant-depth circuits can compute \emph{approximate majority}. That is, they can distinguish  strings with relative Hamming weight $\ge P$ from strings with relative Hamming weight $\le p$ whenever $p<P$ are constants. The proof of the lemma uses circuits for approximate majority.

\begin{proof}[Proof of Lemma \ref{lem:D rho}]
We first construct a distribution over circuits that achieves the goal. Let $a>1$ be a constant that we choose later. Let $n'=\frac{a}{\rho}$.
Let us consider the experiment $E_1$ in which a uniform multi-set $S$ of $[n]$ of size $n'$ is chosen uniformly. (That is $i_1,\ldots,i_{n'}$ are chosen independently and uniformly from $[n]$ and $S$ is the multi-set $\set{i_1,\ldots,i_{n'}}$). Note that for every $x \in \B^n$,
\begin{itemize}
\item If $\weight(x) \ge \rho$ then $\Exp_{S \from E_1}[\sum_{i \in S} x_i] \ge a$.
\item If $\weight(x) \le \rho/10$ then $\Exp_{S \from E_1}[\sum_{i \in S} x_i] \le a/10$.
\end{itemize}
For every multi-set $S \subseteq [n]$ of size $n'$ we consider the circuit $C_S:\B^n \ar \B$ that works as follows:
\begin{itemize}
\item For every choice of $a/5$ elements $j_1,\ldots,j_{a/5}$ in $S$, compute the conjunction of $x_{j_1},\ldots,x_{j_{a/5}}$.
\item Compute the  disjunction of the $(n')^{a/5}$ bits from the previous item.
\end{itemize}
This gives that there exists a constant $c$, that depends on $a$ such that $C_{S}$ is a circuit of size $(1/\rho)^c$ and depth $2$. Furthermore $C_S(x)$ answers $1$ if and only if $\sum_{i \in S}x_i \ge a/5$.

By a (multiplicative) Chernoff bound,\footnote{Specifically, we mean the following version of the Chernoff bound: If $X$ is the sum of $n$ independent variables $X_1,\ldots,X_n \in [0,1]$, and $\Exp(X)=\mu$ then for every $0 \le \delta \le 1$, $\Pr[|X-\mu| \ge \delta \mu] \le 2e^{-\frac{\delta^2 \cdot \mu}{3}}$.}
it follows that there exists a universal constant $\eta>0$ such that for a sufficiently large constant $a$, for every $x \in \B^n$:
\begin{itemize}
\item If $\weight(x) \ge \rho$ then $\Pr_{S \from E_1}[ \sum_{i \in S}x_i \le a/5] \le 2^{-\eta a} \le 1/3$ which implies $\Pr_{S \from E_1}[C_S(x)=1] \ge 2/3$.
\item If $\weight(x) \le \rho/10$ then $\Pr_{S \from E_1}[\sum_{i \in S}x_i \ge a/5] \le 2^{-\eta a} \le 1/3$ which implies $\Pr_{S \from E_1}[C_S(x)=1] \le 1/3$.
\end{itemize}

By taking $t=O(n)$ independent copies of $C_S$ and computing approximate majority, we can reduce the error probability from $1/3$ to $2^{-2n}$, and apply Adleman's argument to obtain a single circuit of size $ \poly(n/\rho)$ and constant depth. Details follow:

For $t=O(n)$ multi-sets $S_1,\ldots,S_t \subseteq [n]$ of size $n'$, we consider the circuit $C_{S_1,\ldots,S_t}(x)$ which for every $i \in [t]$ computes $b_i=C_{S_i}(x)$ and then computes approximate majority (with parameters $p=0.49$ and $P=0.51$) on the string $b=b_1,\ldots,b_t$. Note that each such circuit has constant depth and size $\poly(n/\rho)$.
Let $E_2$ denote the experiment in which the $t$ sets $S_1,\ldots,S_t$ are chosen independently, where each $S_i \subseteq [n]$ is a uniformly chosen multi-set of size $n'$. By a Chernoff bound, it follows that for every $x \in \B^n$:
\begin{itemize}
\item If $\weight(x) \ge \rho$ then $\Pr_{S_1,\ldots,S_t \from E_2}[C_{S_1,\ldots,S_t}(x)=1] \ge 1-2^{-2n}$.
\item If $\weight(x) \le \rho/10$ then $\Pr_{S_1,\ldots,S_t \from E_2}[C_{S_1,\ldots,S_t}(x)=1] \le 2^{-2n}$.
\end{itemize}
By a union bound over all $2^n$ choices of $x \in \B^n$ we obtain that there exist sets $S'_1,\ldots,S'_t$ such that setting $\D^n_{\rho}=C_{S'_1,\ldots,S'_t}$, we obtain the circuit guaranteed in the statement of the theorem.
\end{proof}

We will prove Lemma \ref{lem:hard-core compute exists C} using the following sequence of claims. The overall structure of the argument is similar to the proof of Lemma \ref{lem:exists C}.

\begin{claim}
\label{clm:hard-core compute A}
For every $x \in \B^{\ell}$ and every $\alpha \in \B^t$, there exists a circuit of size $\ell \cdot q \cdot 2^q$ and depth $2$, $A_{x,\alpha}:\B^{2^{\ell'}} \ar \B^{\ell}$ such that for every $h:\B^{\ell'} \ar \B$,
\[ A_{x,\alpha}(h)=\Red^h(x,\alpha). \]
\end{claim}

\begin{proof}
For every $x \in \B^{\ell}$ and every $\alpha \in \B^t$, the function $A_{x,\alpha}(h)$ can be computed by a depth $q$ decision tree, that has outputs of length $\ell$ bits. Each output bit of this function can be computed by a DNF of size $O(q \cdot 2^q)$ and overall, the function can be computed by a depth $2$ circuit of size $\ell \cdot q \cdot 2^q$ as required.
\end{proof}

\begin{claim}
\label{clm:Cg is good}
There exists a universal constant $d$ such that
for every $g:\B^{\ell} \ar \B^{\ell}$, there exists a circuit $C_g:\B^{2^{\ell'}} \ar \B$ of size $\poly(2^q,2^{\ell},2^t)$ and depth $d$ such that the following holds for every $z \in \B^{2^{\ell'}}$:
\begin{itemize}
\item If there exists $\alpha \in \B^t$ such that $\Pr_{x \from U_{\ell}}[\Red^{g^{\pred} \oplus z}(x,\alpha)=g(x)] \ge \rho$ then $C_g(z)=1$.
\item If for all $\alpha \in \B^t$, $\Pr_{x \from U_{\ell}}[\Red^{g^{\pred} \oplus z}(x,\alpha)=g(x)] \le \rho/10$ then $C_g(z)=0$.
\end{itemize}
\end{claim}

\begin{proof}
The circuit $C_g$ will be hardwired with $g$ and $g^{\pred} = \Con(g)$. Upon receiving an input $z \in \B^{2^{\ell'}}$ it will act as follows:
\begin{itemize}
\item Prepare $w=g^{\pred} \oplus z$. (Here we think of $g^{\pred},w,z$ as strings in $\B^{2^{\ell'}}$).
\item For every $x \in \B^{\ell}$ and $\alpha \in \B^t$ compute $A_{x,\alpha}(w)$, and compute:
$b_{x,\alpha} \in \B$ defined by:
\[ b_{x,\alpha}=\left\{\begin{array}{lr}
        0  & A_{x,\alpha}(w) \ne g(x) \\
        1  & A_{x,\alpha}(w) = g(x)
        \end{array}\right. \]
\item For every $\alpha \in \B^t$, let $v_{\alpha}$ denote the $2^{\ell}$ bit long concatenation of all bits $(b_{x,\alpha})_{x \in \B^{\ell}}$ (fixing some order on $x \in \B^{\ell})$, and compute \[b_{\alpha}=\D^{2^{\ell}}_\rho(v_{\alpha}), \]
    where $\D^{2^{\ell}}_{\rho}$ is the circuit guaranteed in Lemma \ref{lem:D rho}.
\item Compute the disjunction of the $2^t$ bits $(b_{\alpha})_{\alpha \in \B^t}$ and output it.
\end{itemize}
It is immediate that the circuit $C_g$ performs the task specified in the lemma. We now explain how to implement the circuit in small size and depth.

The function $g$ can be described using $\ell \cdot 2^{\ell}$ bits. We note that when using the string $g^{\pred}$ to prepare $w$, we only need to have $g_{\pred}$ at coordinates $y \in \B^{\ell'}$ such that there exist $x,\alpha$ such that $A_{x,\alpha}(w)$ depends on $w_y$. As each circuit $A_{x,\alpha}$ is a circuit of size $O(\ell \cdot q \cdot 2^q)$ it depends on at most $O(\ell \cdot q \cdot 2^q)$ input bits. Thus, going over all choices of $x \in \B^{\ell}$ and $\alpha \in \B^t$, $C_g$ only requires $O(2^t \cdot 2^{\ell} \cdot \ell \cdot q \cdot 2^q)$ bits of $g^{\pred}$. Overall, the size of the advice of $C_g$ is $O(2^t \cdot 2^{2\ell} \cdot 2^{2q})$.
The circuit $C_g$ is constant depth by construction, and its size is indeed:
\[ \poly(2^q,2^{\ell},2^t,1/\rho)=\poly(2^q,2^{\ell},2^t), \]
since, by the requirement on $\rho$ we have that $\rho \ge 2^{-\ell}$.
\end{proof}

We will consider the case where $g:\B^{\ell} \ar \B^{\ell}$ is a uniformly chosen function, and will analyze the behavior of $C_g$ on $z \from \BSC^{2^{\ell'}}_{\half-2\eps}$ and on $z \from \BSC^{2^{\ell'}}_{\half}$.
In what follows, let $F_{\ell}$ denote the set of all functions $g:\B^{\ell} \ar \B^{\ell}$.

\begin{claim}\label{clm:hard-to-compute circuit decodes}
$\Pr_{g \from F_{\ell},z \from \BSC^{2^{\ell'}}_{\half-2\eps}}[C_g(z)=1] \ge 0.999$.
\end{claim}

\begin{proof}
Imagine that $g \from F_{\ell}$ is already chosen and fixed, and let $g^{\pred}=\Con(g)$.
By a Chernoff bound, with probability $1-2^{-\frac{1}{3} \cdot \eps^2 \cdot 2^{\ell'}}$ over $z \from \BSC_{\half - 2\eps}^{2^{\ell'}}$, we have that the relative Hamming weight of $z$ is at most $\half-\eps$. This probability is larger than $0.999$ by the requirement that $\eps \ge 2^{-\ell'/3}$ and that $\ell'$ is sufficiently large. Whenever this event occurs,  for $h=g^{\pred} \oplus z$, it holds that:
\[ \Pr_{x \from U_{\ell'}}[h(x)=g^{\pred}(x)] \ge \half+\eps. \]
Therefore, by Definition \ref{dfn:HC compute}
there exists $\alpha \in \B^t$, such that:
    \[ \Pr_{x \from U_{\ell}}[\Red^h(x,\alpha)=g(x)] \ge \rho. \]
By Claim \ref{clm:Cg is good} it follows that whenever this occurs, $C_g(z)=1$, and the claim follows.
\end{proof}

On the other hand, we can show that:

\begin{claim}\label{clm:hard-to-compute circuit does not decode}
$\Pr_{g \from F_{\ell},z \from \BSC^{2^{\ell'}}_{\half}}[C_g(z)=1] \le 0.001$.
\end{claim}

\begin{proof}
The first step of $C_g(z)$ is to prepare $w=g^{\pred} \oplus z$.
However, for  $g \from F_{\ell}$ and $z \from \BSC_{\half}^{2^{\ell'}}$, we have that $w=g^{\pred} \oplus z$ is uniformly chosen, and independent of $g$. This means that the bits $A_{x,\alpha}(w)$ for $x \in \B^{\ell}$ and $\alpha \in \B^t$ are independent of $g$. It follows that for every $\alpha \in \B^t$,  the function $v_{\alpha} \in F_{\ell}$ defined by $v_{\alpha}(x)=(A_{x,\alpha}(w))$, is independent of $g$. Therefore, the probability that $v_{\alpha}$ agrees with $g$ in a $\rho'=\rho/10$ fraction of inputs $x \in \B^{\ell}$ is at most:
\[ {2^{\ell} \choose \rho' \cdot 2^{\ell}} \cdot \frac{1}{2^{\rho' \cdot \ell \cdot 2^{\ell}}} \le \left(\frac{e}{\rho'}\right)^{\rho' \cdot 2^{\ell}} \cdot \left(\frac{1}{2^{\ell}} \right)^{\rho' \cdot 2^{\ell}} \le \frac{1}{2^{\ell/2}}, \]
where the first inequality follows from ${n \choose k} \le (\frac{en}{k})^k$ and the second inequality follows by our requirements on $\rho \ge 2^{-\ell/3}$.

For fixed $w \in \B^{2^{\ell'}}$, the condition that $v_{\alpha}$ agrees with $g$ in a $\rho/10$ fraction of inputs $x \in \B^{\ell}$, can also be phrased as:
\[ \Pr_{x \from U_{\ell}}[\Red^{w}(x,\alpha)=g(x)] \ge \rho/10. \]
It follows that for every $\alpha \in \B^t$, with probability at least $1- 2^{-\ell/2}$ over the choice of $g \from F_{\ell},z \from \BSC^{2^{\ell'}}_{\half}$, we have that:
\[ \Pr_{x \from U_{\ell}}[\Red^{g^{\pred} \oplus z}(x,\alpha)=g(x)] < \rho/10, \]
By a union bound over all $2^t$ choices of $\alpha \in \B^t$, with probability at least $1-2^t \cdot 2^{-\ell/2}$ over the choice of $g \from F_{\ell},z \from \BSC^{2^{\ell'}}_{\half}$, we have that
for all $\alpha \in \B^t$,
\[ \Pr_{x \from U_{\ell}}[\Red^{g^{\pred} \oplus z}(x,\alpha)=g(x)] < \rho/10, \]
which by Claim \ref{clm:Cg is good} implies that $C_g(z)=0$.
Consequently, in order to complete the proof of the claim, it remains to verify that:
\[ 2^t \cdot 2^{-\ell/3} \le 0.001. \]
This follows by the requirements that $t \le 2^{\ell/5}$.
\end{proof}

We are finally ready to prove Lemma \ref{lem:hard-core compute exists C}.

\begin{proof}[Proof of Lemma~\ref{lem:hard-core compute exists C}]
From Claims~\ref{clm:hard-to-compute circuit decodes} and~\ref{clm:hard-to-compute circuit
does not decode}, it follows that:
\begin{align*}
\Pr_{g \from F_{\ell}}\left[\Pr_{z \from \BSC^{2^{\ell'}}_{\half-2\eps}}[C_g(z)=0] > 0.01\right] &\le 0.1,\\
\Pr_{g \from F_{\ell}}\left[\Pr_{z \from \BSC^{2^{\ell'}}_{\half}}[C_g(z)=1] > 0.01 \right] &\le 0.1.
\end{align*}
Thus, by a union bound, there exists $g:\B^\ell \to \B^\ell$ such that
\begin{align*}
\Pr_{z \from \BSC^{2^{\ell'}}_{\half-2\eps}}[C_g(z)=1] &\ge 0.99,\\
\Pr_{z \from \BSC^{2^{\ell'}}_{\half}}[C_g(z)=1] &\le 0.01.
\end{align*}
This completes the proof of the lemma, as $C_g$ satisfies all the other requirements as well.
\end{proof}

\subsection{The case of functions that are hard to invert}
\label{sec:HC invert}

\subsubsection{The model for black-box proofs}
\label{sec:HC invert definition}

In this section we state and explain our model for black-box proofs for hard core predicates, in the setting of functions that are hard to invert. The precise formal definition is given in concise form in Definition \ref{dfn:HC invert}. Below, we provide a detailed explanation for the considerations made in the formal definition. The reader can skip directly to the formal definition if they wish to.

This setting is very similar to the case of functions that are hard to compute, but there are several key differences that we explain below.

\paragraph{Explanation of the model:} Recall that (as explained in Section \ref{sec:intro:HC:invert}) the Goldreich-Levin theorem (stated precisely in Theorem \ref{thm:GL invert}) has the following form:
\begin{itemize}
\item We are given an arbitrary function $f:\B^{\ell} \ar \B^{\ell}$. (Intuitively, it is assumed that $f$ is a one-way function, meaning that it is hard to invert $f$ with success probability $\rho$)
\item There is a specified construction that transforms $f$ into two functions: A ``new one-way function'' $f^{\newowf}:\B^{\ell'} \ar \B^{\ell'}$ and a predicate $f^{\pred}:\B^{\ell'} \ar \B$ for some $\ell'$ related to $\ell$. (Intuitively, we will want to argue that $f^{\pred}$ is a hard-core predicate such that for $x \from U_{\ell'}$, $f^{\pred}(x)$ is hard to compute with success $\half+\eps$ when given $f^{\newowf}(x)$).

    We will model this construction as a map $\Con$ which given $f$ produces a pair of functions $(f^{\newowf},f^{\pred})$. Once again, we place \emph{essentially no limitations} on the map $\Con$ (and in particular do not require that $f^{\newowf},f^{\pred}$ can be efficiently computed if $f$ is). This only makes our results stronger. We remark that unlike the case of functions that are hard to compute, in this setting, we do place some limitations on the construction map $\Con$ (specifically, that the construction of hard-core predicate is nontrivial, in a precise sense explained right after Definition \ref{dfn:HC invert}).

    In the case of Theorem \ref{thm:GL invert}, we have that: $\Con(f)=(f^{\newowf},f^{\pred})$ where $\ell'=2\ell$ and we think of the $\ell'$ bit long input of $(f^{\newowf},f^{\pred})$ as two strings $x,r \in \B^{\ell}$, setting:
    \[ f^{\newowf}(x,r)=(f(x),r), \mbox{and} \]
     \[f^{\pred}(x,r)=\Enc^{\Hadamard}(x)_r = (\sum_{i \in [\ell]} x_i \cdot r_i ) \mod 2. \]
\item We model the proof showing that $f^{\pred}$ is a hard-core predicate, in the following way: The proof is a pair $(\Con,\Red)$ where $\Red^{(\cdot)}$ is an oracle procedure, such that when $\Red^{(\cdot)}$ receives oracle access to an ``adversary'' $h:\B^{\ell'} \ar \B$ that breaks the security of $f^{\pred}$, we have that $\Red^{h}$ breaks the security of $f$.

    It is illustrative to consider the case where $f,f^{\newowf}$ are \emph{permutations}, and with this choice, the model we have introduced so far is identical to the one considered in Section \ref{sec:HC compute} if we set $g=f^{-1}$.

    In the setup of functions that are \emph{hard to invert}, a reduction $\Red$ can potentially want to \emph{compute the function $f$} (as we are implicitly assuming that $f$ is efficiently computable). Many reductions in the cryptographic literature (e.g. from one-way functions to pseudorandom generators) critically rely on this ability, and so, if we want to handle a general case, we should allow the reduction $\Red$ to also \emph{receive oracle access to $f$}, allowing it to compute $f$ on chosen values, if it wants to.

    This means that in the actual definition, $\Red^{(\cdot,\cdot)}$ is an oracle procedure with \emph{two} oracles: It receives oracle access both to $h$ and to $f$.
    More precisely, we require that: for every $f:\B^{\ell} \ar \B^{\ell}$ and for every $h:\B^{\ell'} \ar \B$, such that:
    \[ \Pr_{x \from U_{\ell'}}[h(f^{\newowf}(x))= f^{\pred}(x)] \ge \half+\eps,\]
    it holds that:
    \[\Pr_{x \from U_{\ell}}[\Red^{h,f}(f(x)) \in f^{-1}(f(x))] \ge \rho.\footnote{In fact, we should also allow $h$ to be an oracle procedure $h^{(\cdot)}$ that receives oracle access to $f$. However, as we want to prove lower bounds on black-box proofs, we choose not to do that, as the lower bounds that we prove obviously also rule out this case. This can be intrepretted as saying that the choice of $h$ that we use in our lower bound, does not make calls to $f$.}\]
    Note that this means that even in the case that $f,f^{\newowf}$ are permutations, reductions (in the setting of functions that are hard to invert) are more powerful then reductions (in the setting of functions that are hard to compute) for the function $g=f^{-1}$, and indeed, it is more difficult to prove impossibility results for the case of functions that are hard to invert.
\item As explained in the case of functions that are hard to compute, we are once again aiming to prove a result for circuits (which are allowed to use nonuniform advice) and as in the case of functions that are hard to compute, we will allow the reduction to receive an advice string $\alpha$ of length $t$. (Intuitively, this advice string can depend on $f$ and $h$). This leads to the following strengthening of the requirement above. Namely, we will require that: for every $f:\B^{\ell} \ar \B^{\ell}$ and for every $h:\B^{\ell'} \ar \B$, such that:
    \[ \Pr_{x \from U_{\ell'}}[h(f^{\newowf}(x))= f^{\pred}(x)] \ge \half+\eps,\]
    there exists $\alpha \in \B^t$ such that:
    \[\Pr_{x \from U_{\ell}}[\Red^{h,f}(f(x),\alpha) \in f^{-1}(f(x))] \ge \rho.\]
\item Once again, we make no restrictions on the complexity of the procedure $\Red^{(\cdot,\cdot)}$ except for requiring that it makes at most $q$ queries to each of its two oracles (for some parameter $q$). Our black-box impossibility results will follow from proving lower bounds on $q$.
\end{itemize}

\paragraph{Formal definition: }
Following this discussion, we now give a formal definition.

\begin{definition}[Nonuniform black-box proof for hard-core predicates for hard to invert functions]
\label{dfn:HC invert}
A pair $(\Con,\Red)$ is a \remph{nonuniform black-box proof} for \remph{hard-core predicates for hard to invert functions} with parameters $\ell,\ell',\rho,\eps$, that uses \remph{$q$ queries}, and \remph{$t$ bits of advice} if:
\begin{itemize}
\item $\Con$ is a \emph{construction map} which given a function $f:\B^{\ell} \ar \B^{\ell}$, produces two functions $\Con(f)=(f^{\newowf},f^{\pred})$ such that $f^{\newowf}:\B^{\ell'} \ar \B^{\ell'}$ and $f^{\pred}:\B^{\ell'} \ar \B$.
\item $\Red^{(\cdot,\cdot)}$ is a \emph{reduction}, that is an oracle procedure that given oracle access to functions $h:\B^{\ell'} \ar \B$, and $f:\B^{\ell} \ar \B$, makes at most $q$ queries to each of its two oracles.
\end{itemize}
Furthermore, for every functions $f:\B^{\ell} \ar \B^{\ell}$ and $h:\B^{\ell'} \ar \B$ such that:
    \[ \Pr_{x \from U_{\ell'}}[h(f^{\newowf}(x))=f^{\pred}(x)] \ge \half+\eps, \]
    there exists $\alpha \in \B^t$, such that:
    \[ \Pr_{x \from U_{\ell}}[\Red^{h,f}(f(x),\alpha) \in f^{-1}(f(x))] \ge \rho. \]
\end{definition}

\paragraph{Avoiding trivial constructions: }
We now explain that it is possible to have black-box proofs that are \emph{trivial} and provide hard-core predicates that are hard because of trivial reasons. We need to avoid such trivial constructions if we want to prove interesting limitations.

Specifically, it is easy to obtain hard-core predicates that are hard because information on $f^{\pred}(x)$ is not present in $f^{\newowf}(x)$. Indeed, consider the construction $\Con(f)=(f^{\newowf},f^{\pred})$ with:
\[ f^{\pred}(x)=x_1, \mbox{and} \]
\[ f^{\newowf}(x)=x_2,\ldots,x_{\ell}. \]
In this case, $f^{\pred}$ is a hardcore-predicate, because it is impossible (even for unbounded adversaries) to compute $f^{\pred}(x)$ when given $f^{\newowf}(x)$. This means that for this construction map, there is a reduction that makes $q=0$ queries. Consequently, in order to prove lower bounds, we need to avoid such trivial (and uninteresting) construction maps, and require that for every $f$, there exists a function $\phi_f:\B^{\ell'} \ar \B$ such that for every $x \in \B^{\ell'}$, $\phi_f(f^{\newowf}(x))=f^{\pred}(x)$, meaning that there is information on $f^{\pred}(x)$ in $f^{\newowf}(x)$.

An additional case of a trivial construction that we want to avoid, is the case in which \[ \Hi(f^{\newowf}(U_{\ell'})) < \log(1/\rho). \] Such a construction is uninteresting, because in such a case, if we define the function $\psi_f:\B^{\ell'} \ar \B^{\ell'}$ to output the constant $x \in \B^{\ell}$ such that
\[ \Pr[f^{\newowf}(U_{\ell'})=f^{\newowf}(x)] \ge \rho, \]
(and note that such an $x$ exists if $\Hi(f^{\newowf}(U_{\ell'})) < \log(1/\rho)$) then we get that there is a constant function $\psi_f$ such that $\psi_f$ inverts the function $f^{\newowf}$ with probability $\rho$. Such a construction is uninteresting because in that case $f^{\newowf}$ is obviously not a one-way function.

This leads to the following characterization of nontrivial construction maps, in which we require that $\Con$ avoids these two trivial examples.

\begin{definition}[Nontrivial construction map]
\label{dfn:nontrivial construction}
We say that a construction map $\Con(f)=(f^{\newowf},f^{\pred})$ is $\rho$-\remph{non-trivial} if it satisfies the following two requirements:
\begin{itemize}
\item  For every $f$, the functions $(f^{\newowf},f^{\pred})$ produced by $\Con(f)$ are such that there exists a function $\phi_f:\B^{\ell'} \ar \B$ such that for every $x \in \B^{\ell'}$, $\phi_f(f^{\newowf}(x))=f^{\pred}(x)$.
\item $\Hi(f^{\newowf}(U_{\ell'})) \ge \log(1/\rho)$.
\end{itemize}
We say that a pair $(\Con,\Red)$ is $\rho$-\remph{nontrivial}, if $\Con$ is $\rho$-non-trivial.
\end{definition}

\paragraph{The role of the number of queries, and black-box impossibility results: }
We now explain the role of the parameter $q$ (that measures the number of queries made by $\Red$) and why lower bounds on $q$ translate into black-box impossibility results. This explanation is similar to the one given in Section \ref{sec:HC compute} (with the modifications explained above).

For this purpose, it is illustrative to examine the argument showing that nonuniform black-box proofs yield hard-core predicates: When given a pair $(\Con,\Red)$ that is a nonuniform black-box proof for hard-core predicates for hard to invert functions with parameters $\ell,\ell',\rho,\eps$, that uses $q$ queries, and $t$ bits of advice, we obtain that for any function $f:\B^{\ell} \ar \B^{\ell}$, if there exists a circuit $C':\B^{\ell'} \ar \B$ of size $s'$ such that:
\[ \Pr_{x \from U_{\ell'}}[C'(f^{\newowf}(x))=f^{\pred}(x)] \ge \half+\eps, \]
    then there exists $\alpha \in \B^t$, such that:
    \[ \Pr_{x \from U_{\ell}}[\Red^{C',f}(f(x),\alpha) \in f^{-1}(f(x))] \ge \rho. \]
Note that if the reduction $\Red$ can be implemented by a circuit of size $r$, and the function $f$ can be computed by a circuit of size $m$,
then the circuit $C(y)=\Red^{C',f}(y,\alpha)$ is a circuit of size:
\[ s=r+t+q \cdot m + q \cdot s' \]
that inverts $f$ with success probability $\rho$.

It follows that in a black-box proof, with $q$ queries, and $t$ bits of advice, we get a hard-core theorem that needs to assume that the original function $f$ cannot be inverted by circuits of size $s$, for:
\[ s \ge q +t. \]

\subsubsection{Precise statements of limitations}

Our main result on black-box proofs for hard-core predicates in the setting of functions that are hard to invert is the following theorem.

\begin{theorem}
\label{thm:hard-core invert}
There exist universal constants $\beta>0$ and $c>1$ such that the following holds for any sufficiently large $\ell$ and $\ell'$,
$t \le 2^{\ell/5}$, $\rho \ge 2^{-\ell/5}$, and $\rho \le \beta \cdot \eps^2$. Let $(\Con,\Red)$ be a $\rho$-nontrivial nonuniform black-box proof for hard-core predicates for hard to invert functions with parameters $\ell,\ell',\rho,\eps$, that uses $q$ queries, and $t$ bits of advice. Then
\[ q \ge \frac{1}{\eps^{\beta}}-c(t+\ell). \]
\end{theorem}

We now explain why Theorem \ref{thm:hard-core invert} implies the informal statement made in Theorem \ref{ithm:invert}. This explanation is essentially identical to the one following Theorem \ref{thm:hard-core compute}.

Recall that in Section \ref{sec:HC invert definition} we explained that when using a nonuniform black-box proof to obtain hard-core predicates, we get a hard-core predicate theorem in which $s \ge q+t$.

Theorem \ref{thm:hard-core invert} implies that  it is impossible for such a proof to establish $\eps=s^{-2/\beta}$ (even if $\rho$ is very small). This follows as otherwise, using the fact that $s \ge q+t \ge t$ and $s \geq \ell$, we get that:
\[ q \ge \frac{1}{\eps^{\beta}} -c(t + \ell) \ge s^2 -c(t+\ell)> s,\]
which is a contradiction to $s \ge q+t \ge q$.

Next note that we may assume that $\rho = 2^{-\ell/5}$, as by definition, if $(\Con,\Red)$ is a nonuniform black-box proof for hard-core predicates with parameter $\rho' > \rho $, then it is also a nonuniform black-box proof for hard-core predicates with parameter $\rho$.
Similarly, we may take $\eps=\frac{1}{s^{\omega(1)}}$ to be sufficiently large so that $\eps \geq 2^{-o(\ell)}$ for $s = 2^{-o(\ell)}$.
Under these assumptions, we have that
$\rho \le \beta \cdot \eps^2$, and therefore the parameter setting considered in Theorem \ref{ithm:invert} is impossible to achieve.

\subsubsection{Proof of Theorem \ref{thm:hard-core invert}}

The proof of Theorem \ref{thm:hard-core invert} is similar in structure to the proof of Theorem \ref{thm:hard-core compute} with three main differences:
\begin{itemize}
\item Rather than choosing the initial function uniformly from $F_{\ell}$ (the set of all functions from $\ell$ bits to $\ell$ bits) we will restrict the choice to permutations. This is helpful because for a permutation $f$, the function $f^{-1}$ is well defined, and inverting $f$ (that is producing an element in $f^{-1}(f(x))$ when given $f(x)$) can be thought of as computing $f^{-1}$ that is producing $x$ on input $f(x)$.
\item A more significant difference, is that as explained in detail in Section \ref{sec:HC invert definition}, in the setup of functions that are hard to invert, the reduction $\Red$ has oracle access to $f$ (in addition to oracle access to $h$). This means that it is no longer the case that the answer of the reduction on inputs $x,\alpha$ and oracle $h$ is \emph{determined} by $h,x,\alpha$ (as the answer depends on $f$). Therefore, it is not the case that there exists circuits $A_{x,\alpha}(h)$ that simulate the reduction (as we argued in Claim \ref{clm:hard-core compute A}) and we need to be more careful when showing that for every function $f$, there exists a circuit $C_f(z)$ that is analogous to the circuit guaranteed in Claim \ref{clm:Cg is good}.

    Furthermore, as $\Red$ gets oracle access to $f$, we can no longer claim that when $h$ is independent of $f$, then $\Red$ has no information on $f$. Instead, we use results by Gennaro and Trevisan \cite{GT} showing that an oracle circuit that makes a subexponential number of queries to a random permutation $f$, cannot invert $f$ with high probability.
\item Unlike the case of Section \ref{sec:HC compute}, the distribution that is given as input to $h$ is not necessarily uniform. More precisely, the distribution of $f^{\newowf}(U_{\ell'})$ (on which $h$ needs to predict the hard-core predicate) is not necessarily uniform. However, by the nontriviality condition in Definition \ref{dfn:nontrivial construction} we have that this distribution has high min-entropy, and we need to adjust the argument to hold with this weaker requirement.
\end{itemize}

Theorem \ref{thm:hard-core invert} will follow from the next lemma.

\begin{lemma}
\label{lem:hard-core invert exists C}
There exists a universal constant $d>1$
such that the following holds
for any sufficiently large $\ell$ and $\ell'$, $t \le 2^{\ell/5}$, $\rho \ge 2^{-\ell/5}$ and $\rho \le \frac{\eps^2}{d}$. Let $(\Con,\Red)$ be a $\rho$-nontrivial nonuniform black-box proof for hard-core predicates for hard to invert functions with parameters $\ell,\ell',\rho,\eps$, that uses $q$ queries, and $t$ bits of advice. Then there exists a circuit $C$ of size $\poly(2^q,2^{\ell},2^t)$ and depth $d$ such that:
\begin{itemize}
\item $\Pr_{z \from \BSC^{n}_{\half-2\eps}}[C(z)=1] \ge 0.99$.
\item $\Pr_{z \from \BSC^{n}_{\half}}[C(z)=1] \le 0.01$.
\end{itemize}
\end{lemma}

Once again, just like in the previous section, by reduction to the coin problem, Theorem \ref{thm:hard-core invert} follows from Lemma \ref{lem:hard-core invert exists C}.

\begin{proof}[Proof of Theorem~\ref{thm:hard-core invert}]
The theorem follows directly from Lemma \ref{lem:hard-core invert exists C} and Corollary \ref{cor:coin-problem}, which give that:
\[ \poly(2^q,2^{\ell},2^t) \ge \exp(d \cdot \eps^{-\frac 1 {d-1}}).\]
The statement of Theorem~\ref{thm:hard-core invert} follows by taking the logarithm on both sides and setting $\beta < \frac 1 {d-1}$.
\end{proof}

In the remainder of this section we prove Lemma \ref{lem:hard-core invert exists C}.
Let $(\Con,\Red)$ be a nontrivial nonuniform black-box proof for hard-core predicates for hard to invert functions with parameters $\ell,\ell',\rho,\eps$, that uses $q$ queries, and $t$ bits of advice. Throughout this section we assume that the requirements made in Lemma \ref{lem:hard-core invert exists C} are met.

We will prove Lemma \ref{lem:hard-core invert exists C} using the following sequence of claims. The proof uses the same structure as the proof of Lemma \ref{lem:hard-core compute exists C} however, the reduction is now more powerful as it has oracle access to the function $f$, and the setting is more general as the distribution $f^{\newowf}(U_{\ell})$ (on which the oracle is judged) is not necessarily uniform.
In this setting, there is no direct analog of Claim \ref{clm:hard-core compute A}, and instead we prove an analog of Claim \ref{clm:Cg is good} directly.

\begin{claim}
\label{clm:Cf is good}
There exists a universal constant $d$ such that
for every permutation $f:\B^{\ell} \ar \B^{\ell}$, there exists a circuit $C_f:\B^{2^{\ell'}} \ar \B$ of size $\poly(2^q,2^{\ell},2^t)$ and depth $d$ such that the following holds for every $z \in \B^{2^{\ell'}}$:
\begin{itemize}
\item If there exists $\alpha \in \B^t$ such that $\Pr_{x \from U_{\ell}}[\Red^{\phi_f \oplus z,f}(f(x),\alpha) = x] \ge \rho$ then $C_f(z)=1$.
\item If for all $\alpha \in \B^t$, $\Pr_{x \from U_{\ell}}[\Red^{\phi_f \oplus z,f}(f(x),\alpha)=x] \le \rho/10$ then $C_f(z)=0$.
\end{itemize}
\end{claim}

\begin{proof}
The circuit $C_f$ will be hardwired with $f$ and $\phi_f$ (where $\phi_f$ is the function whose existence is guaranteed for $f$ by the nontriviality condition in Definition \ref{dfn:nontrivial construction}). Upon receiving an input $z \in \B^{2^{\ell'}}$ it will act as follows:
\begin{itemize}
\item Prepare $w=\phi_f \oplus z$. (Here we think of $\phi_f,z,w$ as  strings in $\B^{2^{\ell'}}$).
\item For every $x \in \B^{\ell}$ and $\alpha \in \B^t$ compute $\Red^{w,f}(f(x),\alpha)$, and compute
$b_{x,\alpha} \in \B$ defined by:
\[ b_{x,\alpha}=\left\{\begin{array}{lr}
        0  & \Red^{w,f}(f(x),\alpha) \ne x \\
        1  & \Red^{w,f}(f(x),\alpha) = x
        \end{array}\right. \]
\item For every $\alpha \in \B^t$, let $v_{\alpha}$ denote the $2^{\ell}$ bit long concatenation of all bits $(b_{x,\alpha})_{x \in \B^{\ell}}$ (fixing some order on $x \in \B^{\ell})$, and compute \[b_{\alpha}=\D^{2^{\ell}}_\rho(v_{\alpha}), \]
        where $\D^{2^{\ell}}_{\rho}$ is the circuit guaranteed in Lemma \ref{lem:D rho}.
\item Compute the disjunction of the $2^t$ bits $(b_{\alpha})_{\alpha \in \B^t}$ and output it.
\end{itemize}
It is immediate that the circuit $C_f$ performs the task specified in the lemma.
We now explain how to implement the circuit in small size and depth.

The string $f$ can be described by $\ell \cdot 2^{\ell}$ bits. We note that when using the string $\phi_f$ to prepare $w$, we only need to have $\phi_f$ at coordinates $y \in \B^{\ell'}$ such that there exist $x,\alpha$ such that $\Red^{w,f}(f(x),\alpha)$ depends on $w_y$. As on every pair $(f(x),\alpha)$ the reduction $\Red^{w,f}(f(x),\alpha)$ makes at most $q$ queries to its oracle $w$, it can depend on at most $2^q$ choices of $y \in \B^{\ell'}$.
Thus, going over all choices of $x \in \B^{\ell}$ and $\alpha \in \B^t$, $C_g$ only requires $O(2^t \cdot 2^{\ell} \cdot 2^q)$ bits of $\phi_f$. Note that any query that the reduction makes to its oracle $f$, is a constant that $C_f$ has hardwired (because $C_f$ is hardwired with $f$).
Overall, the size of the advice of $C_f$ is $O(2^t \cdot 2^{\ell} \cdot 2^q)$.
The circuit $C_f$ is constant depth by construction, and its size is indeed:
\[ \poly(2^q,2^{\ell},2^t,1/\rho)=\poly(2^q,2^{\ell},2^t), \]
by the requirement that $\rho \ge 2^{-\ell}$.
\end{proof}

We will consider the case where $f:\B^{\ell} \ar \B^{\ell}$ is a uniformly chosen permutation, and will analyze the behavior of $C_f$ on $z \from \BSC^{2^{\ell'}}_{\half-2\eps}$ and on $z \from \BSC^{2^{\ell'}}_{\half}$.
In what follows, let $\Pi_{\ell}$ denote the set of all permutations $f:\B^{\ell} \ar \B^{\ell}$.

\begin{claim}
\label{clm:hard-to-invert circuit decodes}
$\Pr_{f \from \Pi_{\ell},z \from \BSC^{2^{\ell'}}_{\half-2\eps}}[C_f(z)=1] \ge 0.999$.
\end{claim}

\begin{proof}
Imagine that $f \from \Pi_{\ell}$ is already chosen and fixed. Let $(f^{\newowf},f^{\pred})=\Con(f)$, and let $\phi_f:\B^{\ell'} \ar \B$ be the function guaranteed by the fact that $\Con$ is nontrivial.
We now consider the additional experiment of choosing $z \from \BSC^{2^{\ell'}}_{\half-2\eps}$.
Let $h:\B^{\ell'} \ar \B$ be defined by $h=\phi_f \oplus z$. Our goal is to show that with probability at least $0.999$ over choosing $z \from \BSC^{2^{\ell'}}_{\half-2\eps}$, we have that:
\begin{equation}
\label{eqn:newowf and pred}
\Pr_{x \from U_{\ell'}}[h(f^{\newowf}(x))=f^{\pred}(x)] \ge \half+\eps.
\end{equation}
This is because whenever $f,z$ satisfy the condition above, then by the properties of $\Red$, we have that there exists $\alpha \in \B^t$ such that:
\[ \Pr_{x \from U_{\ell}}[\Red^{h,f}(f(x),\alpha)=x] \ge \rho. \]
which in turn by Claim \ref{clm:Cf is good} implies that $C_f(z)=1$.

By definition, for every $x \in \B^{\ell'}$, $\phi_f(f^{\newowf}(x))=f^{\pred}(x)$. Consequently, the event \[\set{h(f^{\newowf}(x))=f^{\pred}(x)}\] that appears in (\ref{eqn:newowf and pred}) can be expressed as \[\set{h(f^{\newowf}(x))=\phi_f(f^{\newowf}(x))}.\] As $h=\phi_f \oplus z$, this event can also be expressed as \[ \set{z(f^{\newowf}(x))=0}.\]
Thus, in order to prove the claim, it is sufficient to prove that with probability at least $0.999$ over choosing $z \from \BSC^{2^{\ell'}}_{\half-2\eps}$, we have that:
\[ \Pr_{x \from U_{\ell'}}[z(f^{\newowf}(x))=0] \ge \half+\eps. \]
Note that $Y=f^{\newowf}(U_{\ell'})$ is not necessarily uniform. For every $y \in \B^{\ell'}$, we define $p_y=\Pr[Y=y]$.
By the nontriviality of $\Con$ we have that $\Hi(Y) \ge \log(1/\rho)$, which means that for every $y \in \B^{\ell'}$, $p_y \le \rho$.
Thus, in order to conclude the proof, it is sufficient to show that:
\[ \Pr_{z \from \BSC^{2^{\ell'}}_{\half-2\eps}}[\sum_{y \in \B^{\ell'}} p_y \cdot z_y < \half-\eps] \ge 0.999. \]
(This can be thought of as a ``weighted version'' of Hamming weight in which the $p_y$ are not all the same).
When $z \from \BSC^{2^{\ell'}}_{\half-2\eps}$, the $2^{\ell'}$ random variables $x_y=p_y \cdot z_y$ (one for each choice of $y \in \B^{\ell'})$ are independent, and lie in the interval $[0,\rho]$. We have that:
\[\Exp _{z \from \BSC^{2^{\ell'}}_{\half-2\eps}}[\sum_{y \in \B^{\ell'}} p_y \cdot z_y] = \half-2\eps.\]
We can apply a Chernoff bound to bound the probability of deviation from the expectation and obtain that:
\[ \Pr_{z \from \BSC^{2^{\ell'}}_{\half-2\eps}}[\sum_{y \in \B^{\ell'}} p_y \cdot z_y < \half-\eps] \le e^{-\Omega(\eps^2/\rho)}, \]
By our requirement that $\rho \leq \beta \cdot \eps^2$ for a sufficiently small $\beta >0$, we get that the probability is indeed larger than $0.999$.
\end{proof}

On the other hand, we can show that:

\begin{claim}
\label{clm:Cf fails}
$\Pr_{f \from \Pi_{\ell},z \from \BSC^{2^{\ell'}}_{\half}}[C_f(z)=1] \le 0.001$.
\end{claim}

In order to prove Claim \ref{clm:Cf fails} we will use the following result by Gennaro and Trevisan \cite{GT}:

\begin{theorem}[\cite{GT}]
\label{thm:GT}
For sufficiently large $\ell$, for every oracle procedure $P^{(\cdot)}$ that makes at most $2^{\ell/5}$ queries to its oracle, and accepts inputs $x \in \B^{\ell}$ and $\alpha \in \B^{t}$ for $t \le 2^{\ell/5}$, it holds that:
\[\Pr_{f \from \Pi_{\ell}}[\exists \alpha \in \B^t:\ \Pr_{x \from U_{\ell}}[P^f(f(x),\alpha)=x] \ge 2^{-\ell/5}] \le 2^{-2^{\ell/2}} \]
\end{theorem}

We now prove Claim \ref{clm:Cf fails}.
\begin{proof}[Proof of Claim \ref{clm:Cf fails}]
The first step of $C_f(z)$ is to prepare $w=\phi_f \oplus z$.
However, for  $f \from \Pi_{\ell}$ and $z \from \BSC_{\half}^{2^{\ell'}}$, we have that $w=\phi_f \oplus z$ is uniformly chosen, and independent of $f$.

Therefore, for any choice of advice $\alpha \in\B^t$ oracle access to $w$ does not help the reduction $\Red(\cdot,\alpha)$ to invert $f$. More precisely, Let $P^{(\cdot)}(x,\alpha)$ be an implementation of $\Red^{(\cdot,\cdot)}$ where whenever $\Red$ makes a query to its first oracle $h$, the query is answered by a fresh uniform random bit.
We have that:
\[ \Pr_{f \from \Pi_{\ell},z \from \BSC^{2^{\ell'}}_{\half}}[C_f(z)=1] \]
\[\ \le \Pr_{f \from \Pi_{\ell},z \from \BSC^{2^{\ell'}}_{\half}}[\exists \alpha \in \B^t:\ \Pr_{x \from U_{\ell}}[\Red^{\phi_f \oplus z,f}(f(x),\alpha)=x] \ge 2^{-\ell/5}]  \]
\[\ = \Pr_{f \from \Pi_{\ell},w \from \BSC^{2^{\ell'}}_{\half}}[\exists \alpha \in \B^t:\ \Pr_{x \from U_{\ell}}[\Red^{w,f}(f(x),\alpha)=x] \ge 2^{-\ell/5}]  \]
\[ \le \Pr_{f \from \Pi_{\ell}}[\exists \alpha \in \B^t:\ \Pr_{x \from U_{\ell}}[P^f(f(x),\alpha)=x] \ge 2^{-\ell/5}]  \]
\[ \le 2^{-2^{\ell/2}} \le 0.001, \]
where the first inequality follows from Claim \ref{clm:Cf is good} and assumption that $\rho \geq 2^{-\ell/5}$, and
where the penultimate inequality follows from Theorem \ref{thm:GT}.
\end{proof}

The proof of Lemma \ref{lem:hard-core invert exists C} now follows  from Claims~\ref{clm:hard-to-invert circuit decodes} and~\ref{clm:Cf fails} in exactly the same way as in the end of the previous section. Specifically:

\begin{proof}[Proof of Lemma~\ref{lem:hard-core invert exists C}]
From Claims~\ref{clm:hard-to-invert circuit decodes} and~\ref{clm:Cf fails}, it follows that:
\begin{align*}
\Pr_{f \from \Pi_{\ell}}\left[\Pr_{z \from \BSC^{2^{\ell'}}_{\half-2\eps}}[C_f(z)=0] > 0.01\right] &\le 0.1,\\
\Pr_{f \from \Pi_{\ell}}\left[\Pr_{z \from \BSC^{2^{\ell'}}_{\half}}[C_f(z)=1] > 0.01 \right] &\le 0.1.
\end{align*}
Thus, by a union bound, there exists $f:\B^\ell \to \B^\ell$ such that
\begin{align*}
\Pr_{z \from \BSC^{2^{\ell'}}_{\half-2\eps}}[C_f(z)=1] &\ge 0.99,\\
\Pr_{z \from \BSC^{2^{\ell'}}_{\half}}[C_f(z)=1] &\le 0.01.
\end{align*}
This completes the proof of the lemma, as $C_f$ satisfies all the other requirements as well.
\end{proof}

\section*{Acknowledgment}

We are grateful to Ilan Newman for participating in early stages of this research and for many helpful discussions.

Noga Ron-Zewi was partially supported by ISF grant 735/20. Ronen Shaltiel was partially supported by ISF grant 1628/17. Nithin Varma was partially supported by ISF grant 497/17,
and Israel PBC Fellowship for Outstanding Postdoctoral Researchers from India and China.
\bibliographystyle{alpha}
\bibliography{references}

\appendix

\section{Proof of Lemma \ref{lem:high-ent-guess}}
\label{sec:app:high-ent-guess}

In this section we prove Lemma \ref{lem:high-ent-guess}, restated below.

\highentguess*

In this proof we use the following notation. For two distributions $X,Y$ over $\B^n$, we say that they are $\eps$-close if for every event $A \subseteq \B^n$, $|\Pr[X \in A] - \Pr[Y \in A]| \le \eps$. We will also use Shannon's entropy   given by $$H(X) = \sum_x \Pr(X=x) \cdot \log\left( \frac {1} {\Pr(X=x)} \right),$$ and the following statement of Pinsker's lemma:

\begin{lemma}[Pinsker's lemma]
If $X$ is a distribution over $\B^n$ and $H(X) \ge n-\eps$, then $X$ is $\sqrt{\eps}$-close to $U_n$.
\end{lemma}

\begin{proof}[Proof of lemma \ref{lem:high-ent-guess}]
By the requirements on $M$, we have that $H(M) \ge k-k^{0.99}$.
The Shannon entropy function satisfies $H(M_1) + \ldots + H(M_k) \ge H(M_1,\ldots,M_k)$ and therefore:
\[ H(M_1) + \ldots + H(M_k) \ge k-k^{0.99} \]
It follows that:
\[ \Exp_{i \from [k]} H(M_i) = 1- k^{-0.01}. \]
By Markov's inequality, for every $c$, the fraction of $i \in [k]$ such that $H(M_i) < 1-c \cdot k^{-0.01}$ is less than $1/c$.
By Pinsker's lemma, for every $i$ such that $H(M_i) \ge 1-c \cdot k^{-0.01}$, we have that $M_i$ is $\sqrt{c \cdot k^{-0.01}}$-close to $U_1$.
Therefore,
\[ \Pr_{m \from M,i \from [k]}[D(i)=m_i] \le \half+\sqrt{c \cdot k^{-0.01}} + 1/c \le 0.5001, \]
for a sufficiently large constant $c$, and sufficiently large $k$.
\end{proof}

\section{Proof of Corollary \ref{cor:coin-problem}}
\label{sec:app:coin problem}

In this Section we prove Corollary \ref{cor:coin-problem}, restated below.

\coin*

The proof is by reduction to Theorem \ref{thm:coin-problem}.

\begin{proof}[Proof of Corollary \ref{cor:coin-problem}]
For $\eps'=\Theta(\eps$), given a circuit $C$ that satisfies:
\begin{itemize}
\item $\Pr_{z \from \BSC^{n}_{\half-\eps'}}[C(z)=1] \ge 0.99$,
\item $\Pr_{z \from \BSC^{n}_{\half}}[C(z)=1] \le 0.01$.
\end{itemize}
We will show the existence of a circuit $C'$ that satisfies:
\begin{itemize}
\item $\Pr_{z \from \BSC^{n}_{\half-\eps}}[C'(z)=1] \ge 0.9$,
\item $\Pr_{z \from \BSC^{n}_{\half+\eps}}[C'(z)=1] \le 0.1$.
\end{itemize}
We will start by constructing a randomized circuit $C'$ which upon receiving input $x \in \B^n$, for every $i \in [n]$ independently, $C'$ replaces input bit $x_i$ by zero with probability $p=\frac{2\eps}{1+2\eps}$, let $x'_i$ denote the obtained bit, and let $C'(x)=C(x')$.
The choice of $p$ is made so that for every $i$:
\begin{itemize}
\item If $x \from \BSC^n_{\half+\eps}$ then $x' \from \BSC^n_{\half}$.
\item if $x \from \BSC^n_{\half-\eps}$ then $x' \from \BSC^n_{\half-\eps'}$ for $\eps'=\eps + \eps \cdot \frac{1-2\eps}{1+2\eps} = \Theta(\eps)$.
\end{itemize}
It follows that if $C$ distinguishes between $\BSC^n_{\half-\eps'}$ and $\BSC^n_{\half}$ then $C'$ satisfies:
\begin{itemize}
\item $\Pr_{z \from \BSC^{n}_{\half-\eps}}[C'(z)=1] \ge 0.99$,
\item $\Pr_{z \from \BSC^{n}_{\half+\eps}}[C'(z)=1] \le 0.01$.
\end{itemize}
This gives that:
\[ \Pr_{z \from \BSC^{n}_{\half-\eps}}[C'(z)=1] - \Pr_{z \from \BSC^{n}_{\half+\eps}}[C'(z)=1] \ge 0.98, \]
where the probability in the expressions above is also over the randomness of $C'$.
Therefore, there exists a fixing of the random coins of $C'$ which achieves this gap of $0.98$ and, hence, satisfies the requirements on $C'$ (this follows because for numbers $0 \le p \le P < 1$ that satisfy $P-p \ge 0.98$, it holds that: $P \ge 0.9$ and $p \le 0.1$).
We note that the size and depth of $C'$ are bounded by the size and depth of $C$ respectively, and the corollary follows.
\end{proof}

\end{document}